\documentclass[12pt]{amsart}
\usepackage{amssymb,latexsym}
\usepackage{epsfig}
\newtheorem{theorem}{Theorem}

\newtheorem{remark}{Remark}
\newcommand{\Cbb}{\mathbb{C}}
\newcommand{\Rbb}{\mathbb{R}}
\newcommand{\imag}{\text{Im\,}}

\newcommand{\dive}{\,\text{div}\,}
\newcommand{\sech}{\,\text{sech}\,}

\newcommand{\calt}{\hat{\mathcal{T}}}
\newcommand{\lcalt}{\hat{t}}
\newcommand{\cald}{\mathcal{D}}
\newcommand{\tilt}{\hat{\mathfrak{T}}}
\newcommand{\crx}{\hat{a}^{\dag}_{x}}
\newcommand{\anx}{\hat{a}_{x}}
\newcommand{\cry}{\hat{a}^{\dag}_{y}}
\newcommand{\any}{\hat{a}_{y}}
\newcommand{\Op}{\text{Op}^{W}_{h}}
\newcommand{\Hcaln}{\mathcal{H}_{n}}

\begin{document}

\title{Manipulation of Semiclassical Photon States}
\author{Michael VanValkenburgh}
\address{UCLA Department of Mathematics, Los Angeles, CA 90095-1555, USA}
\email{mvanvalk@ucla.edu}

\begin{abstract}
Gabriel F. Calvo and Antonio Pic\'{o}n defined a class of
operators, for use in quantum communication, that allows arbitrary
manipulations of the three lowest two-dimensional Hermite-Gaussian
modes
$\mathcal{H}_{\mathcal{T}}=\{|0,0\rangle,|1,0\rangle,|0,1\rangle\}$.
Our paper continues the study of those operators, and our results
fall into two categories. For one, we show that the generators of
the operators have infinite deficiency indices, and we explicitly
describe all self-adjoint realizations. And secondly we
investigate semiclassical approximations of the propagators. The
basic method is to start from a semiclassical Fourier integral
operator ansatz and then construct approximate solutions of the
corresponding evolution equations. In doing so, we give a complete
description of the Hamilton flow, which in most cases is given by
elliptic functions. We find that the semiclassical approximation
behaves well when acting on sufficiently localized initial
conditions, for example, finite sums of semiclassical
Hermite-Gaussian modes, since near the origin the Hamilton
trajectories trace out the bounded components of elliptic curves.
\end{abstract}

\maketitle

\section{Introduction}

In a series of papers, Gabriel F. Calvo, Antonio Pic\'{o}n, and
co-authors studied the manipulation of single-photon states for
purposes of quantum communication, from theory to experimental
design \cite{R:CalvoPiconBagan}, \cite{R:CalvoPicon},
\cite{R:CalvoPiconZam}. They demonstrated the limitations of the
metaplectic operators when acting on spatial transverse-field
modes, and they showed how one can overcome those limitations with
a different family of transformations \cite{R:CalvoPicon}. In
particular, they considered the subspace of states spanned by the
three lowest Hermite-Gaussian (transverse spatial) modes
$\mathcal{H}_{\mathcal{T}}=\{|0,0\rangle,|1,0\rangle,|0,1\rangle\}$
and found a class of operators that would allow arbitrary
manipulations of such modes. However, their class of operators
includes ``non-Gaussian transformations'', which are not
metaplectic operators. Thus the question arises of what properties
of metaplectic operators may be extended, at least partially, to
their non-Gaussian transformations. This could lead to the
construction of the associated optical system \cite{R:CalvoPicon}.

The theory of metaplectic operators is often presented in terms of the Stone-von Neumann theorem (see, for example, the book of G. B. Folland \cite{R:Folland}). Rather than attempt to extend the Stone-von Neumann theorem to non-Gaussian transformations, in this paper we provide an alternative approach, based on the theory of semiclassical Fourier integral operators. We find that a Fourier integral operator ansatz provides approximate solutions to the evolution equations of Calvo and Pic\'{o}n, and in the process we will determine the canonical transformations associated with the non-Gaussian transformations, in the sense of Egorov's theorem. In fact, the canonical transformations are given by the Hamilton flow of the semiclassical Weyl symbols of Calvo and Pic\'{o}n's [now semiclassical] differential operators. The Weyl symbols are dependent on the semiclassical parameter $h$, so the Hamilton flow is also $h$-dependent. We take this point of view essentially for three reasons: (1) The operators may be exactly reconstructed from their ($h$-dependent) Weyl symbols. (2) The flow leaves invariant a disc of radius $\sim\sqrt{h}$, which closes up as $h\rightarrow 0$. In the context of laser physics, $$h=\frac{1}{2}w_{0}^{2},$$ where $w_{0}$ denotes the radius of the laser beam's waist. Outside this disc of radius $\sim w_{0}$, the flow goes to infinity in finite time. This is because the flow propagates along elliptic curves having two components, one bounded and one unbounded.  (See Section \ref{S:FIOnonGauss} and the Appendix.) And (3) the appropriate version of Egorov's theorem has an error of order $\mathcal{O}(h^{2})$, rather than the more typical $\mathcal{O}(h)$. (See Section \ref{S:Egorov}.) Moreover, in Section \ref{S:Concl} we give examples suggesting that we have creation and propagation of singularities along the $h$-dependent flow.

In this paper we find that there are fundamental difficulties with
the generators of the non-Gaussian transformations. First of all,
we show that, although they are clearly symmetric, they have
infinite deficiency indices. We are however able to describe all
self-adjoint realizations in terms of certain boundary conditions
at infinity. And there are difficulties with constructing
semiclassical approximations of the unitary groups (the
non-Gaussian transformations). We can still find approximate
solutions to the evolution equations, but the associated
symplectic transformations can blow up in finite time, as they are
described in terms of the Weierstrass $\wp$-function, which, as is
well-known, has a double pole in every period. However, we show
that the semiclassical propagator behaves reasonably well when
acting on sufficiently phase-space localized initial conditions,
for example, finite sums of semiclassical Hermite-Gaussian modes.
This is because, as mentioned above, the Hamilton flow leaves
invariant a disc of radius $\sim\sqrt{h}$ centered at the origin.

In Section~\ref{S:CalvoPiconReview} we review the non-Gaussian transformations of Calvo and Pic\'{o}n and show that the generators have infinite deficiency indices. We then describe all self-adjoint realizations in terms of certain boundary conditions at infinity. In Section~\ref{S:SCformalism} we introduce the semiclassical version of Calvo and Pic\'{o}n's work, so that we can solve their evolution equations approximately, in the semiclassical regime. We illustrate this method in Sections \ref{S:WarmUp} and \ref{S:Gyrator} by first taking simpler generators and using semiclassical Fourier integral operators to find \emph{exact} solutions to the associated evolution equations, corresponding to ``Gaussian transformations''. We use the same method for the more difficult non-Gaussian transformations in Sections \ref{S:FIOnonGauss} and \ref{S:tilt4}, but in that case we only have approximate solutions to the evolution equations. In Section \ref{S:Egorov} we discuss Egorov's theorem and give a property of non-Gaussian transformations that is analogous to a property of metaplectic transformations. We give concluding remarks in Section \ref{S:Concl}, and we include an appendix for relevant facts regarding elliptic curves.\\

\noindent\textbf{Notation:} We write $D_{x}=\frac{1}{i}\frac{\partial}{\partial x}$ and similarly for the other variables. Also, we will use the semiclassical Weyl quantization of a symbol $\sigma$, defined by
\begin{equation*}
   \Op(\sigma)(x,\xi)=(2\pi h)^{-n}\iint e^{\frac{i(x-y)\xi}{h}}\sigma\left(\frac{x+y}{2},\xi\right)f(y)\,dy\,d\xi.
\end{equation*}
We will mostly be in dimension $n=1$ or $n=2$, and we will mostly
deal with polynomials $\sigma$ (hence resulting in semiclassical
differential operators with polynomial coefficients).

\vspace{12pt}

\section{The Non-Gaussian Transformations of Calvo and Pic\'{o}n}\label{S:CalvoPiconReview}

In this section we recapitulate the recent work of Calvo and
Pic\'{o}n \cite{R:CalvoPicon}. They introduced the following eight
generators acting on Hermite-Gaussian modes:
$$\calt_{1}=\frac{1}{2}(\crx\any+\cry\anx),$$
$$\calt_{2}=-\frac{i}{2}(\crx\any-\cry\anx),$$
$$\calt_{3}=\frac{1}{2}(\crx\anx-\cry\any),$$
$$\calt_{4}=\frac{1}{2}(\crx+\anx-\crx\crx\anx-\crx\anx\anx-\crx\cry\any-\anx\cry\any),$$
$$\calt_{5}=-\frac{i}{2}(\crx-\anx-\crx\crx\anx+\crx\anx\anx-\crx\cry\any+\anx\cry\any),$$
$$\calt_{6}=\frac{1}{2}(\cry+\any-\cry\cry\any-\cry\any\any-\cry\crx\anx-\any\crx\anx),$$
$$\calt_{7}=-\frac{i}{2}(\cry-\any-\cry\cry\any+\cry\any\any-\cry\crx\anx+\any\crx\anx),$$
$$\calt_{8}=\frac{1}{2\sqrt{3}}[-2+3(\crx\anx+\cry\any)],$$
defined in terms of the creation and annihilation operators $\crx=\frac{1}{\sqrt{2}}\left(x-\frac{\partial}{\partial x}\right)$ and $\anx=\frac{1}{\sqrt{2}}\left(x+\frac{\partial}{\partial x}\right)$, respectively, and similarly for the $y$-variable. We note first of all that $\calt_{4}$ and $\calt_{5}$ are unitarily related by the Fourier transform $\mathcal{F}$:
\begin{equation*}
    \calt_{5}=\mathcal{F}\circ\calt_{4}\circ\mathcal{F}^{-1},
\end{equation*}
and that $\calt_{6}$ and $\calt_{7}$ are obtained from $\calt_{4}$ and $\calt_{5}$ by simply interchanging the variables $x$ and $y$.

These generators, within the subspace generated by the lowest three Hermite-Gaussian modes
$\mathcal{H}_{\calt}=\{|0,0\rangle,|1,0\rangle,|0,1\rangle\}$, obey the $SU(3)$ algebra\footnote{We sum over $c$, which is only relevant here for $(a,b)=(4,5)$ and $(a,b)=(6,7)$.}
$$[\calt_{a},\calt_{b}]=if_{abc}\calt_{c}$$
($a,b,c=1,2,\ldots 8$), where the only nonvanishing (up to permutations) structure constants $f_{abc}$ are given by
$$f_{123}=1, \quad f_{147}=f_{165}=f_{246}=f_{257}=f_{345}=f_{376}=1/2,\quad\text{and }f_{458}=f_{678}=\sqrt{3}/2.$$
We note that the triad of generators
$$\Gamma_{1}\equiv\{\calt_{1},\calt_{2},\calt_{3}\}$$
gives a $SU(2)$ group that conserves the mode order. The remaining two $SU(2)$ groups are formed by the triads $$\Gamma_{2}\equiv\{\calt_{4},\calt_{5},(\calt_{3}+\sqrt{3}\calt_{8})/2\}$$ and $$\Gamma_{3}\equiv\{\calt_{6},\calt_{7},(-\calt_{3}+\sqrt{3}\calt_{8})/2\}.$$
Unitary operators $\hat{U}_{\Gamma_{1}}$ generated by the first triad give rise to superpositions between the two modes $|1,0\rangle$ and $|0,1\rangle$, leaving invariant the fundamental mode $|0,0\rangle$. Unitarities $\hat{U}_{\Gamma_{2}}$ and $\hat{U}_{\Gamma_{3}}$, generated by the second and third triad, produce superpositions between the two modes $|0,0\rangle$ and $|1,0\rangle$ (leaving invariant $|0,1\rangle$), or the modes $|0,0\rangle$ and $|0,1\rangle$ (leaving invariant $|1,0\rangle$), respectively. For reference, we state the operations for $\hat{U}_{\Gamma_{2}}$:
\begin{align*}
    e^{it\calt_{4}}|0,0\rangle&=\cos(\frac{t}{2})|0,0\rangle+i\sin(\frac{t}{2})|1,0\rangle\\
    e^{it\calt_{5}}|0,0\rangle&=\cos(\frac{t}{2})|0,0\rangle+\sin(\frac{t}{2})|1,0\rangle\\
    e^{it(\frac{1}{2}(\calt_{3}+\sqrt{3}\calt_{8}))}|0,0\rangle&=e^{-\frac{it}{2}}|0,0\rangle\\
    e^{it\calt_{4}}|1,0\rangle&=\cos(\frac{t}{2})|1,0\rangle+i\sin(\frac{t}{2})|0,0\rangle\\
    e^{it\calt_{5}}|1,0\rangle&=\cos(\frac{t}{2})|1,0\rangle-\sin(\frac{t}{2})|0,0\rangle\\
    e^{it(\frac{1}{2}(\calt_{3}+\sqrt{3}\calt_{8}))}|1,0\rangle&=e^{\frac{it}{2}}|1,0\rangle.
\end{align*}
The corresponding operations for $\hat{U}_{\Gamma_{1}}$ and $\hat{U}_{\Gamma_{3}}$ are similar.

\vspace{12pt}

However, a difficulty arises when attempting to extend to larger
subspaces of $L^{2}(\Rbb^{2})$: starting from the smallest natural
domain, the domain of [finite] linear combinations of
Hermite-Gaussian modes, the symmetric operators $\calt_{4}$,
$\calt_{5}$, $\calt_{6}$, and $\calt_{7}$ all have multiple
self-adjoint realizations.

\vspace{12pt}

We first prove that the deficiency indices of $\calt_{4}$ are both infinity. For this we consider how $\calt_{4}$ acts on the basis of two-dimensional Hermite-Gaussian modes, which we write either as $|m,n\rangle$ or as $\psi_{m}^{n}$, with $m$ being the mode order in the $x$-variable and $n$ being the mode order in the $y$-variable. We have that
\begin{equation}\label{E:caltonmn}
    \calt_{4}\psi_{m}^{n}=\beta_{m}^{n}\psi_{m+1}^{n}+\beta_{m-1}^{n}\psi_{m-1}^{n},
\end{equation}
where
\begin{equation*}
    \beta_{m}^{n}=\frac{1}{2}\sqrt{m+1}(1-m-n).
\end{equation*}

We take the domain of $\calt_{4}$ (as an unbounded linear operator) to be the subspace $D(\calt_{4})$ of [finite] linear combinations of Hermite-Gaussian modes. This domain is dense in $L^{2}(\Rbb^{2})$, and $\calt_{4}$ is clearly symmetric with this domain. And it is easy to check that the domain of the adjoint operator is
\begin{equation}\label{E:domadj}
    D(\calt_{4}^{*})=\left\{g=\sum_{m,n=0}^{\infty}g_{m}^{n}\psi_{m}^{n}\in L^{2}(\Rbb^{2});\, \sum_{m,n=0}^{\infty}|\beta_{m}^{n}g_{m+1}^{n}+\beta_{m-1}^{n}g_{m-1}^{n}|^{2}<\infty\right\},
\end{equation}
and that, for $g=\sum_{m,n=0}^{\infty}g_{m}^{n}\psi_{m}^{n}\in D(\calt_{4}^{*})$,
\begin{equation*}
    \calt_{4}^{*}g=\sum_{m,n=0}^{\infty}(\beta_{m}^{n}g_{m+1}^{n}+\beta_{m-1}^{n}g_{m-1}^{n})\psi_{m}^{n}.
\end{equation*}
And we occasionally find it convenient to use the formal operator on sequences given by
\begin{equation*}
    \lcalt_{4}:\quad (g)_{m}^{n}\mapsto (\lcalt_{4}g)_{m}^{n}:=\beta_{m}^{n}g_{m+1}^{n}+\beta_{m-1}^{n}g_{m-1}^{n}.
\end{equation*}

\vspace{12pt}

Two somewhat degenerate cases occur when $n=0$ and $n=1$, so we treat these separately. For $n=0$, we get
\begin{equation*}
    \left(
    \begin{matrix}
    0&\frac{1}{2}&0&0&0&0&\cdots\\
    \frac{1}{2}&0&0&0&0&0\\
    0&0&0&-\frac{\sqrt{3}}{2}&0&0&\cdots\\
    0&0&-\frac{\sqrt{3}}{2}&0&-2&0\\
    0&0&0&-2&0&-\frac{3}{2}\sqrt{5}\\
    0&0&0&0&-\frac{3}{2}\sqrt{5}&\ddots&\ddots\\
    0&0&0&0&0&\ddots&\ddots\\
    \vdots&&\vdots&&\vdots&&\ddots
    \end{matrix}
    \right),
\end{equation*}
and for $n=1$, we get
\begin{equation*}
    \left(
    \begin{matrix}
    0&0&0&0&0&\cdots\\
    0&0&-\frac{1}{\sqrt{2}}&0&0\\
    0&-\frac{1}{\sqrt{2}}&0&-\sqrt{3}&0&\cdots\\
    0&0&-\sqrt{3}&0&\ddots\\
    0&0&0&\ddots&\ddots&\ddots\\
    \vdots&&\vdots&&\ddots&\ddots
    \end{matrix}
    \right).
\end{equation*}
In all other cases we have
\begin{equation*}
    \left(
    \begin{matrix}
    0&\beta_{0}^{n}&0&0&\cdots\\
    \beta_{0}^{n}&0&\beta_{1}^{n}&0\\
    0&\beta_{1}^{n}&0&\beta_{2}^{n}\\
    0&0&\beta_{2}^{n}&0&\ddots\\
    \vdots&&&\ddots&\ddots
    \end{matrix}
    \right)
\end{equation*}
where $\beta_{m}^{n}<0$ for all $m$ and $n$.

Hence $\calt_{4}$ acts like a Jacobi matrix for any fixed $n$, fixing the mode order in the $y$-variable, and $\calt_{4}$ may be decomposed accordingly. To be precise, we consider the orthogonal projection operator onto the $n^{th}$ mode in $y$:
\begin{equation}\label{E:Pnproj}
    P_{n}:\,\,L^{2}(\Rbb^{2})\ni f\mapsto \sum_{m=0}^{\infty}\langle f|\psi_{m}^{n}\rangle\psi_{m}^{n}.
\end{equation}
Then (\ref{E:caltonmn}) shows that $\calt_{4}$ is decomposed
according to $$L^{2}(\Rbb^{2})=\Hcaln\oplus\Hcaln^{\perp},$$ where
$\Hcaln$ and $\Hcaln^{\perp}$ are by definition the kernels of
$1-P_{n}$ and $P_{n}$, respectively. That is, we have (see, for
example, the book of Kato \cite{R:Kato})
\begin{equation*}
    P_{n}D(\calt_{4})\subset D(\calt_{4}),\quad \calt_{4}\Hcaln\subset\Hcaln,\quad\text{and}\quad \calt_{4}\Hcaln^{\perp}\subset\Hcaln^{\perp}.
\end{equation*}

\vspace{12pt}

Now, to prove that $\calt_{4}$ is not essentially self-adjoint on $D(\calt_{4})$, for any fixed $n$ we consider the equations
\begin{equation*}
    \begin{aligned}
    (\lcalt_{4}u)_{m}&=\beta_{m-1}^{n}u_{m-1}+\beta_{m}^{n}u_{m+1}\\
    &=zu_{m}
    \end{aligned}
\end{equation*}
for $z\in\Cbb$. This is a recursion relation for $u_{m+1}$ in terms of $u_{m}$ and $u_{m-1}$. In the special case when $n=0$ we simply take $u_{0}=u_{1}=0$; then the solution is uniquely determined by the initial value $u_{2}$. If $n=1$ we simply take $u_{0}$=0, so that the solution is uniquely determined by the initial value $u_{1}$. In all other cases the solution is completely determined by $u_{-1}=0$ and the initial value $u_{0}$. Hence $u_{m}$ is the initial condition multiplied by a polynomial in $z$, which we write as $P_{m}(z)$.

Let $z\in\Cbb\backslash\Rbb$ and denote by $N_{\overline{z}}$ the orthogonal complement of $\mathcal{R}(\calt_{4}-\overline{z}I)$, called the deficiency subspace of the operator $\calt_{4}$. This subspace is precisely the subspace of solutions of the equation $$\calt_{4}^{*}\phi=z\phi.$$ Considering $D(\calt_{4}^{*})$, stated in (\ref{E:domadj}), we see that $N_{\overline{z}}$ is the subspace given by solutions of the difference equation $$(\lcalt_{4}u)_{m}=zu_{m},\qquad u_{-1}=0$$ (or the appropriate modified version for $n=0$ and $n=1$. For the remainder we always write the initial condition as $u_{0}$, for convenience). Hence for fixed $n$ the deficiency subspace is at most one-dimensional; moreover, it is nonzero if and only if
\begin{equation*}
    (u_{m})=(u_{0}P_{m}(z))\in \ell^{2}(\mathbb{N}_{0}),
\end{equation*}
that is, if
\begin{equation*}
    \sum_{m=0}^{\infty}|P_{m}(z)|^{2}<\infty.
\end{equation*}
Non-self-adjointness is then a consequence of the following theorem about the infinite Jacobi matrix
\begin{equation*}
    L=
    \left(
    \begin{matrix}
    b_{0}&a_{0}&0&0&\cdots\\
    a_{0}&b_{1}&a_{1}&0\\
    0&a_{1}&b_{2}&a_{2}\\
    0&0&a_{2}&b_{3}&\ddots\\
    \vdots&&&\ddots&\ddots
    \end{matrix}
    \right)
\end{equation*}
where $a_{j}>0$ and $b_{j}\in\Rbb$ for all $j$. We can reduce to this case if we multiply $\calt_{4}$ by $-1$. This theorem is from Berezanskii's book (\cite{R:Berezanskii}, p.507), where one may find a beautiful and elementary proof.

\vspace{12pt}

\begin{theorem}\label{T:Berez}
    Assume that $|b_{j}|\leq C$ (j=0,1,\ldots), $a_{j-1}a_{j+1}\leq a_{j}^{2}$ beginning with some $j$, and
    \begin{equation}\label{E:BerezHyp}
        \sum_{j=0}^{\infty}\frac{1}{a_{j}}<\infty.
    \end{equation}
    Then the operator $L$ is not self-adjoint.
\end{theorem}

\vspace{12pt}

It is elementary to show that the hypotheses are satisfied for $a_{m}=-\beta_{m}^{n}$ and $b_{m}=0$, for any fixed $n$. Hence for any fixed $n$ the deficiency indices of the resulting Jacobi matrix are both $1$, and so, summing over $n$, the deficiency indices for $\calt_{4}$ are both infinity.

\vspace{12pt}

As for the operator $\calt_{5}$, we may either use a slight modification of the above methods, or we may simply use the fact that the Fourier transform intertwines $\calt_{4}$ and $\calt_{5}$. Then, simply by interchanging the roles of the variables $x$ and $y$, we see that the operators $\calt_{6}$ and $\calt_{7}$ also have infinite deficiency indices.

\vspace{12pt}

The symmetric cubic operators $\calt_{4}$, $\calt_{5}$, $\calt_{6}$ and $\calt_{6}$ of Calvo and Pic\'{o}n do however admit self-adjoint extensions. Moreover, using the results of Allahverdiev \cite{R:Alla} we may explicitly classify all of the self-adjoint extensions in terms of certain boundary conditions at infinity. We begin by studying the restriction of $\calt_{4}$ to the subspace $\Hcaln$ given by the projection $P_{n}$ as in (\ref{E:Pnproj}). We simplify notation by then omitting ``$n$''; in particular, $\psi_{m}$ now denotes the $(m,n)^{th}$ Hermite-Gaussian mode, and $\beta_{m}$ will also lose the superscript ``$n$''.

For $$g=\sum_{m=0}^{\infty} g_{m}\psi_{m}\in D(\calt_{4})\cap\Hcaln$$
and $$h=\sum_{m=0}^{\infty} h_{m}\psi_{m}\in D(\calt_{4})\cap\Hcaln,$$ we denote by $[g,h]$ the sequence with components $[g,h]_{m}$ given by
\begin{equation*}
    [g,h]_{m}=\beta_{m}(g_{m}\overline{h}_{m+1}-g_{m+1}\overline{h}_{m}).
\end{equation*}
We then have Green's formula:
\begin{equation*}
    \sum_{m=0}^{M}\{(\lcalt_{4}g)_{m}\overline{h}_{m}-g_{m}(\lcalt_{4}\overline{h})_{m}\}=-[g,h]_{M}
\end{equation*}
Since the sequences $(g)_{m}$, $(h)_{m}$, $(\lcalt_{4}g)_{m}$, and $(\lcalt_{4}h)_{m}$ are all in $\ell^{2}(\mathbb{N}_{0})$, we then have that the limit
\begin{equation*}
    [g,h]_{\infty}=\lim_{M\rightarrow\infty}[g,h]_{M}
\end{equation*}
exists and is finite. Hence
\begin{equation*}
    \langle\lcalt_{4}g|h\rangle-\langle g|\lcalt_{4}h\rangle=-[g,h]_{\infty}.
\end{equation*}

Now (still for a fixed $n$) we let
\begin{equation*}
    u=(u_{m}) \quad\text{and}\quad v=(v_{m})
\end{equation*}
be the solutions of
\begin{equation*}
    \beta_{m-1}y_{m-1}+\beta_{m}y_{m+1}=0 \quad (m\geq 1)
\end{equation*}
satisfying the boundary conditions
\begin{equation*}
    u_{0}=1,\quad u_{1}=0, \quad v_{0}=0, \quad\text{and}\quad v_{1}=\frac{1}{\beta_{0}}.
\end{equation*}
(We make the appropriate trivial modifications for $n=0$ and $n=1$.) We have that $u,v\in D(\calt_{4}^{*})\cap\Hcaln$; in fact,
\begin{equation*}
    (\lcalt_{4} u)_{m}=0\,\,\,\text{for all $m$}, \quad (\lcalt_{4} v)_{0}=1,\quad\text{and}\quad (\lcalt_{4}v)_{m}=0,\,\,m\geq 1.
\end{equation*}

With this set-up we have the following results of Allahverdiev \cite{R:Alla}:

\vspace{12pt}

\begin{theorem}\label{T:AllaClos}\textbf{\emph{(Allahverdiev \cite{R:Alla}.)}}
    The domain of the closure of $\calt_{4}$ restricted to $\Hcaln$ consists precisely of those $f\in D(\calt_{4}^{*})\cap\Hcaln$ satisfying the boundary conditions
    \begin{equation*}
        [f,u]_{\infty}=[f,v]_{\infty}=0.
    \end{equation*}
\end{theorem}

\vspace{12pt}

For $f\in D(\calt_{4}^{*})\cap\Hcaln$ we now define
\begin{equation*}
    \Gamma_{1}f=[f,v]_{\infty}
\end{equation*}
and
\begin{equation*}
    \Gamma_{2}f=[f,u]_{\infty}.
\end{equation*}
Then, in the precise sense of unbounded operators on a Hilbert space,

\vspace{12pt}

\begin{theorem}\textbf{\emph{(Allahverdiev \cite{R:Alla}.)}}
    The operators $\Gamma_{1}$, $\Gamma_{2}$ are (complex-valued, symmetric, linearly independent) boundary values of $\calt_{4}$ restricted to $\Hcaln$.
\end{theorem}

\vspace{12pt}

And now that we have appropriate boundary values, we have the following well-known description of all self-adjoint extensions:

\vspace{12pt}

\begin{theorem}\textbf{\emph{(Allahverdiev \cite{R:Alla}.)}}
    When restricted to $\Hcaln$, every self-adjoint extension $\calt_{4}^{h}$ of $\calt_{4}$ is determined by the equality
    \begin{equation*}
        \calt_{4}^{h}f=(\lcalt_{4}f)
    \end{equation*}
    on the functions $f\in D(\calt_{4}^{*})\cap\Hcaln$ satisfying the boundary conditions
    \begin{equation}\label{E:boundcond}
        [f,v]_{\infty}-h[f,u]_{\infty}=0
    \end{equation}
    for $h\in\Rbb\cup\{\infty\}$. For $h=\infty$ the condition (\ref{E:boundcond}) should be replaced by $[f,u]_{\infty}=0$. Conversely, for an arbitrary $h\in\Rbb\cup\{\infty\}$, the boundary condition (\ref{E:boundcond}) determines a self-adjoint extension on $\Hcaln$.
\end{theorem}

\vspace{12pt}

 \begin{remark}
 Allahverdiev \cite{R:Alla} additionally considers maximal dissipative and accretive extensions of infinite Jacobi matrices. These correspond to $h\in\Cbb$ such that $\imag h\geq 0$ and $\imag h\leq 0$, respectively. He also gives applications to scattering theory.
\end{remark}

\vspace{12pt}

It remains to consider $\calt_{4}$ as acting on the entire space
\begin{equation*}
    L^{2}(\Rbb^{2})=\bigoplus_{n=0}^{\infty}\Hcaln.
\end{equation*}
Let $\calt_{4}^{n}$ denote the restriction of $\calt_{4}$ to $\Hcaln$, that is, the operator in $\Hcaln$ with $D(\calt_{4}^{n})=D(\calt_{4})\cap\Hcaln$ such that $\calt_{4}^{n}f=\calt_{4}f\in\Hcaln$. As we have just shown, the closure of $\calt_{4}^{n}$ is obtained by extending to the space of $f\in D(\calt_{4}^{n*})$ satisfying $\Gamma_{1}^{n}f=\Gamma_{2}^{n}f=0$. We now prove the analogous result for the larger space.

\vspace{12pt}

\begin{theorem}\label{T:t4closure}
    The closure of $\calt_{4}$ is obtained by extending to the space
    \begin{equation*}
        \cald_{cl}=\{f\in D(\calt_{4}^{*});\,\Gamma_{1}^{n}(P_{n}f)=\Gamma_{2}^{n}(P_{n}f)=0\,\,\forall n\}.
    \end{equation*}
\end{theorem}

\begin{proof}
    The closure of $\calt_{4}$ is of course the adjoint of $\calt_{4}^{*}$, so we are to show that $\cald_{cl}$ is equal to
    \begin{equation*}
        D(\calt_{4}^{**})=\{F\in L^{2}(\Rbb^{2});\, \exists H\in L^{2}(\Rbb^{2})\,\,\text{such that}\,\, \langle\calt_{4}^{*}g|F\rangle=\langle g|H\rangle\,\forall g\in D(\calt_{4}^{*})\}.
    \end{equation*}
    Moreover, we note that $D(\calt_{4}^{**})\subset D(\calt_{4}^{*})$ since $\calt_{4}^{**}$ and $\calt_{4}$ have the same adjoint and since $\calt_{4}^{**}$ is automatically symmetric.

    Now let $F\in\calt_{4}^{*}$ and write
    \begin{equation*}
        F=\sum F_{m}^{n}\psi_{m}^{n},\quad g=\sum g_{m}^{n}\psi_{m}^{n},\quad \text{ and } H=\sum H_{m}^{n}\psi_{m}^{n}.
    \end{equation*}
    Then we wish to find those precise conditions on $F$ such that there is some $H\in L^{2}(\Rbb^{2})$ with the property that
    \begin{equation*}
        \langle \calt_{4}^{*}g| F\rangle = \langle g|H\rangle \qquad \forall g\in D(\calt_{4}^{*}).
    \end{equation*}
    Simply by restricting to $g\in D(\calt_{4})$, we see that it is necessary to have
    \begin{equation*}
        H_{m}^{n}=\beta_{m-1}^{n}F_{m-1}^{n}+\beta_{m}^{n}F_{m+1}^{n}.
    \end{equation*}
    So we must find the conditions on $F$ such that
    \begin{equation*}
        \langle \calt_{4}^{*}g| F\rangle=\langle g|\calt_{4}^{*}F\rangle \qquad \forall g\in D(\calt_{4}^{*}).
    \end{equation*}

    We recall that
    \begin{equation*}
        [g,F]_{M}^{n}=\sum_{m=0}^{M}\left(g_{m}^{n}\overline{(\lcalt_{4}F)_{m}^{n}}-(\lcalt_{4}g)_{m}^{n}\overline{F_{m}^{n}}\right),
    \end{equation*}
    so we see that
    \begin{equation*}
        \sum_{n=0}^{\infty}[g,F]_{M}^{n}
    \end{equation*}
    converges absolutely. Moreover, by the dominated convergence theorem,
    \begin{equation*}
        \lim_{M\rightarrow\infty}\sum_{n=0}^{\infty}[g,F]_{M}^{n}=\sum_{n=0}^{\infty}\lim_{M\rightarrow\infty}[g,F]_{M}^{n}\equiv \sum_{n=0}^{\infty}[g,F]_{\infty}^{n}.
    \end{equation*}
    Hence
    \begin{equation*}
        \sum_{n=0}^{\infty}[g,F]_{\infty}^{n}=\langle g|\calt_{4}^{*}F\rangle-\langle\calt_{4}^{*}g|F\rangle\qquad\forall g\in D(\calt_{4}^{*}).
    \end{equation*}
    So $D(\calt_{4}^{**})$ is precisely the set of $F\in D(\calt_{4}^{*})$ such that
    \begin{equation*}
        \sum_{n=0}^{\infty}[g,F]_{\infty}^{n}=0\qquad\forall g\in D(\calt_{4}^{*}).
    \end{equation*}
    But then this is equivalent to
    \begin{equation*}
        [g,F]_{\infty}^{n}\equiv [g,P_{n}(F)]_{\infty}^{n}=0 \qquad\forall n,\,\,\forall g\in D(\calt_{4}^{*})\cap \Hcaln,
    \end{equation*}
    which in turn, as shown by Allahverdiev \cite{R:Alla}, is equivalent to
    \begin{equation*}
        \Gamma_{1}^{n}(P_{n}F)=\Gamma_{2}^{n}(P_{n}F)=0\qquad\forall n.
    \end{equation*}
    So indeed we have $\cald_{cl}=D(\calt_{4}^{**})$.
\end{proof}

\vspace{12pt}

We have also seen that self-adjoint extensions of $\calt_{4}^{n}$ are in one-to-one correspondence with boundary conditions of the form
\begin{equation*}
    \Gamma_{1}^{n}(f)-h_{n}\Gamma_{2}^{n}(f)=0
\end{equation*}
for $h_{n}\in\Rbb\cup\{\infty\}$. Now we define
\begin{equation*}
    \cald_{h}=\{f\in D(\calt_{4}^{*});\,\Gamma_{1}^{n}(P_{n}f)-h_{n}\Gamma_{2}^{n}(P_{n}f)=0 \,\,\forall n\},
\end{equation*}
and we denote the extension of $\calt_{4}$ to this domain as $\calt_{4}^{h}$. We next show that all self-adjoint extensions of $\calt_{4}$ are of this form.

\vspace{12pt}

\begin{theorem}
    Every self-adjoint extension $\calt_{4}^{h}$ of $\calt_{4}$ is determined by extending the domain to a set of the form
    \begin{equation*}
        \cald_{h}=\{f\in D(\calt_{4}^{*});\,\Gamma_{1}^{n}(P_{n}f)-h_{n}\Gamma_{2}^{n}(P_{n}f)=0 \,\,\forall n\},
    \end{equation*}
    where $h=(h_{n})_{n=0}^{\infty}$ is an arbitrary sequence with $h_{n}\in\Rbb\cup\{\infty\}$, and by the rule
    \begin{equation*}
        \calt_{4}^{h}f=(\lcalt_{4}f)\equiv\calt_{4}^{*}f.
    \end{equation*}
    For $h_{n}=\infty$ the condition should be replaced by $\Gamma_{2}^{n}(P_{n}f)=0$. Conversely, for an arbitrary sequence $h=(h_{n})_{n=0}^{\infty}$ with $h_{n}\in\Rbb\cup\{\infty\}$, the set $\cald_{h}$ determines a self-adjoint extension of $\calt_{4}$.
\end{theorem}

\begin{proof}
    We begin by proving the converse: we take an arbitrary sequence $h$ with $h_{n}\in\Rbb\cup\{\infty\}$ and prove that the extension $\calt_{4}^{h}$ to the domain $\cald_{h}$ is self-adjoint.

    We first show that $\cald_{h}$ is symmetric, that is, we show that
    \begin{equation*}
        \langle\calt_{4}^{*}g|f\rangle=\langle g|\calt_{4}^{*}f\rangle\qquad\forall f,g\in\cald_{h}.
    \end{equation*}
    As in the proof of Theorem \ref{T:t4closure}, we use the identity
    \begin{equation*}
        \sum_{n=0}^{\infty}[g,f]_{\infty}^{n}=\langle g|\calt_{4}^{*}f\rangle-\langle\calt_{4}^{*}g|f\rangle.
    \end{equation*}
    But now we may use the identity
    \begin{equation*}
        [g,f]_{M}^{n}=[g,u^{n}]_{M}^{n}[\overline{f},v^{n}]_{M}^{n}-[g,v^{n}]_{M}^{n}[\overline{f},u^{n}]_{M}^{n},
    \end{equation*}
    where $u^{n}$ and $v^{n}$ are the functions occurring in the definitions of $\Gamma_{1}^{n}$ and $\Gamma_{2}^{n}$. Hence, in the limit,
    \begin{equation*}
        [g,f]_{\infty}^{n}=(\Gamma_{2}^{n}P_{n}g)(\Gamma_{1}^{n}P_{n}\overline{f})-(\Gamma_{1}^{n}P_{n}g)(\Gamma_{2}^{n}P_{n}\overline{f}).
    \end{equation*}
    Since $g,f$ are in $\cald_{h}$, we can use the identities
    \begin{equation*}
        \Gamma_{1}^{n}(P_{n}f)-h_{n}\Gamma_{2}^{n}(P_{n}f)=0
    \end{equation*}
    and
    \begin{equation*}
        \Gamma_{1}^{n}(P_{n}g)-h_{n}\Gamma_{2}^{n}(P_{n}g)=0
    \end{equation*}
    to see that $[g,f]_{\infty}^{n}=0$ for all $n$. Hence $\calt_{4}^{h}$ is a symmetric operator.

    To show that $\calt_{4}^{h}$ is a closed operator, we suppose that
    \begin{equation*}
        \cald_{h}\ni f_{k}\rightarrow f\qquad\text{in }L^{2}(\Rbb^{2})
    \end{equation*}
    and
    \begin{equation*}
        \calt_{4}^{h}f_{k}\rightarrow F\qquad\text{in }L^{2}(\Rbb^{2}).
    \end{equation*}
    Then clearly, for all $n$,
    \begin{equation*}
        P_{n}f_{k}\rightarrow P_{n}f\qquad\text{in }\Hcaln
    \end{equation*}
    and
    \begin{equation*}
        P_{n}\calt_{4}^{h}f_{k}=\calt_{4}^{h}P_{n}f_{k}\rightarrow P_{n}F\qquad\text{in }\Hcaln.
    \end{equation*}
    But since $\calt_{4}^{h,n}$, the restriction of $\calt_{4}^{h}$ to $\Hcaln$, is closed (as is $\calt_{4}^{*}$), we then have that $P_{n}f\in\cald_{h}\cap\Hcaln$ (hence $f\in\cald_{h}$) and that $\calt_{4}^{h}P_{n}f=P_{n}\calt_{4}^{h}f=P_{n}F$ (hence $\calt_{4}^{h}f=F$). So $\calt_{4}^{h}$ is a closed operator.

    Next we show that the deficiency indices of $\calt_{4}^{h}$ are both zero. Suppose $u_{\pm}\in D(\calt_{4}^{h*})$ are such that $u_{\pm}\in\text{Ker}(\calt_{4}^{h*}\mp i)$. Then $P_{n}u_{\pm}\in\text{Ker}(\calt_{4}^{h,n*}\mp i)$, and, since $\calt_{4}^{h,n}$ is self-adjoint, $P_{n}u_{\pm}=0$ for all $n$. Hence $u_{\pm}=0$. So the deficiency indices of $\calt_{4}^{h}$ are both zero, and $\calt_{4}^{h}$ is a closed operator. Hence $\calt_{4}^{h}$ is self-adjoint.

    For the other direction, let $\calt_{4}^{e}$, with domain $\cald_{e}$, be some self-adjoint extension of $\calt_{4}$, and let $\calt_{4}^{e,n}$ be its restriction to $\Hcaln$. To show that $\calt_{4}^{e,n}$ is closed, we suppose that
    \begin{equation*}
        \cald_{e}\cap\Hcaln\ni f_{k}\rightarrow f\qquad\text{in }\Hcaln
    \end{equation*}
    and that
    \begin{equation*}
        \calt_{4}^{e,n}f_{k}\rightarrow F\qquad\text{in }\Hcaln.
    \end{equation*}
    Then, in particular,
    \begin{equation*}
        \cald_{e}\ni f_{k}\rightarrow f\qquad\text{in }L^{2}(\Rbb^{2})
    \end{equation*}
    and
    \begin{equation*}
        \calt_{4}^{e}f_{k}\rightarrow F\qquad\text{in }L^{2}(\Rbb^{2}).
    \end{equation*}
    Since $\calt_{4}^{e}$ is closed, we have $f\in\cald_{e}\cap\Hcaln$ and $\calt_{4}^{e,n}f=F$, which proves that $\calt_{4}^{e,n}$ is a closed operator. And since the deficiency indices of $\calt_{4}^{e}$ are both zero, it is clear that the deficiency indices of $\calt_{4}^{e,n}$ are both zero. Hence $\calt_{4}^{e,n}$ is self-adjoint.

    Now $\calt_{4}^{e,n}$ is a self-adjoint extension of $\calt_{4}^{n}$, so the results of Allahverdiev cited above show that $\calt_{4}^{e,n}$ must be given by a boundary condition of the form
    \begin{equation*}
        \Gamma_{1}^{n}(f)-h_{n}\Gamma_{2}^{n}(f)=0
    \end{equation*}
    for some $h_{n}\in\Rbb\cup\{\infty\}$. Hence we have $\cald_{e}\subset\cald_{h}$, that is,
    \begin{equation*}
        \calt_{4}^{e}\subset\calt_{4}^{h},
    \end{equation*}
    and since both operators are self-adjoint, we then have $\calt_{4}^{e}=\calt_{4}^{h}$. This completes the proof of the theorem.
\end{proof}

\vspace{12pt}

 We have now categorized all self-adjoint extensions of the operator $\calt_{4}$, so we are presented with the basic question: which is the ``right'' extension? We expect the physics of the problem to dictate the appropriate extension, for which we may return to the original work of Calvo and Pic\'{o}n \cite{R:CalvoPiconBagan},\cite{R:CalvoPicon},\cite{R:CalvoPiconZam}. They state that the cubic generators, $\calt_{4}$, $\calt_{5}$, $\calt_{6}$, and $\calt_{7}$, can be implemented with passive optical elements having higher-than-first-order aberrations (nonquadratic refractive surfaces) \cite{R:CalvoPicon}. Perhaps the physics of the apparatus will determine the ``right'' self-adjoint extension.

\vspace{12pt}

\section{The Semiclassical Formalism}\label{S:SCformalism}

In the case of Gaussian transformations, one may use the theory of metaplectic operators (sometimes presented in terms of the Stone-von Neumann theorem) to deduce the underlying canonical transformations. However, non-Gaussian transformations are not metaplectic operators, so the arguments must be modified. Calvo and Pic\'{o}n then ask if the Stone-von Neumann theorem can be extended to the case of the above cubic generators, for then ``this would enable one to find the explicit form of the symplectic transform and thus the construction of the associated optical system'' \cite{R:CalvoPicon}. In the following sections we take a different approach: we approximate the non-Gaussian transformations by semiclassical Fourier integral operators, by starting with the semiclassical Fourier integral operator ansatz and then by solving the resulting eikonal equation and transport equations. This, along with Egorov's theorem (stated in Section \ref{S:Egorov}), justifies the claim that the underlying canonical transformation is the Hamilton flow associated to the evolution equation. As a first step, in this section we put the work of Calvo and Pic\'{o}n in the framework of semiclassical analysis.

We wish to construct approximate solutions to the evolution equations associated to the cubic generators $\calt_{4}$, $\calt_{5}$, $\calt_{6}$, and $\calt_{7}$, approximate in the sense that we will work in the semiclassical regime. For this we start with the two-dimensional semiclassical ground state
\begin{equation*}
    |0,0\rangle=(\pi h)^{-\frac{1}{2}}e^{-\frac{1}{2h}(x^{2}+y^{2})}.
\end{equation*}
To get the other semiclassical Hermite functions, we apply the creation operators
\begin{equation*}
    \crx=(2h)^{-\frac{1}{2}}\left(x-h\frac{\partial}{\partial x}\right)\qquad\text{and}\qquad \cry=(2h)^{-\frac{1}{2}}\left(y-h\frac{\partial}{\partial y}\right).
\end{equation*}
We also have the corresponding annihilation operators
\begin{equation*}
    \anx=(2h)^{-\frac{1}{2}}\left(x+h\frac{\partial}{\partial x}\right)\qquad\text{and}\qquad \any=(2h)^{-\frac{1}{2}}\left(y+h\frac{\partial}{\partial y}\right),
\end{equation*}
so that
\begin{equation*}
    [\anx,\crx]=1\quad\text{and}\quad [\any,\cry]=1.
\end{equation*}
The [normalized] semiclassical Hermite functions are then given by
\begin{equation*}
    |m,n\rangle=(m!n!)^{-\frac{1}{2}}\crx{}^{m}\cry{}^{n}|0,0\rangle.
\end{equation*}

We again let
\begin{equation*}
   \calt_{4}=\frac{1}{2}(\crx+\anx-\crx\crx\anx-\crx\anx\anx-\crx\cry\any-\anx\cry\any)
\end{equation*}
but now where the creation and annihilation operators are semiclassical. In this context it is more convenient to introduce the semiclassical differential operator
\begin{equation*}
    \tilt_{4}:=-2^{\frac{1}{2}}h^{\frac{3}{2}}\calt_{4}.
\end{equation*}
For reference, we note that
\begin{equation}\label{E:tiltexp}
    \tilt_{4}=\frac{1}{2}x\left((hD_{x})^{2}+(hD_{y})^{2}+x^{2}+y^{2}\right)-\frac{5}{2}hx-\frac{1}{2}ih^{2}D_{x},
\end{equation}
and we note that the semiclassical Weyl symbol of $\tilt_{4}$ is
\begin{equation*}
    p_{4}(x,y,\xi,\eta;h)=\frac{1}{2}x(x^{2}+y^{2}+\xi^{2}+\eta^{2})-\frac{5}{2}hx.
\end{equation*}
We then have the permutation properties similar to those in Section \ref{S:CalvoPiconReview}. Explicitly, for $\hat{U}_{\Gamma_{2}}$ we have
\begin{align*}
    e^{it\tilt_{4}/h}|0,0\rangle&=\cos\left(2^{-1/2}h^{1/2}t\right)|0,0\rangle-i\sin\left(2^{-1/2}h^{1/2}t\right)|1,0\rangle\\
    e^{it\tilt_{5}/h}|0,0\rangle&=\cos\left(2^{-1/2}h^{1/2}t\right)|0,0\rangle-\sin\left(2^{-1/2}h^{1/2}t\right)|1,0\rangle\\
    e^{it(\frac{1}{2}(\tilt_{3}+\sqrt{3}\tilt_{8}))/h}|0,0\rangle&=e^{it2^{-1/2}h^{1/2}}|0,0\rangle\\
    e^{it\tilt_{4}/h}|1,0\rangle&=\cos\left(2^{-1/2}h^{1/2}t\right)|1,0\rangle-i\sin\left(2^{-1/2}h^{1/2}t\right)|0,0\rangle\\
    e^{it\tilt_{5}/h}|1,0\rangle&=\cos\left(2^{-1/2}h^{1/2}t\right)|1,0\rangle+\sin\left(2^{-1/2}h^{1/2}t\right)|0,0\rangle\\
    e^{it(\frac{1}{2}(\tilt_{3}+\sqrt{3}\tilt_{8}))/h}|1,0\rangle&=e^{-it2^{-1/2}h^{1/2}}|1,0\rangle.
\end{align*}
The operations for $\hat{U}_{\Gamma_{1}}$ and $\hat{U}_{\Gamma_{3}}$ are similar. Hence, in this choice of scale, we have permutations of the three lowest Hermite-Gaussian modes when $$t:\, 0\mapsto\frac{\pi}{\sqrt{2h}}.$$

\vspace{12pt}

\section{Warm-Up: The FIO Representation of a Gaussian Transformation}\label{S:WarmUp}

For non-Gaussian transformations, the approximate representation by Fourier integral operators will be somewhat complicated, so we begin with a simpler situation: the case of a Gaussian transformation. For the sake of concreteness, we restrict our attention to the semiclassical differential operator
\begin{equation*}
P(h)=-h(\crx\cry+\anx\any)=-h^{2}\frac{\partial^{2}}{\partial x\partial y}-xy,
\end{equation*}
though all ten generators of $\hat{U}(S)$, the metaplectic representation of $Sp(4,\Rbb)$ (see \cite{R:CalvoPicon}), may be treated in the same way.  Our goal is then to solve the evolution equation
\begin{equation}\label{E:evoleqKx}
    \begin{cases}
        (hD_{t}+P)u=0\\
        u|_{t=0}=v.
    \end{cases}
\end{equation}

Following the method outlined in the book of Grigis and Sj\"{o}strand (\cite{R:GrigisSjostrand}, p.129), we try
\begin{equation*}
    u=U_{t}v(x,y)=(2\pi h)^{-2}\iint e^{\frac{i}{h}\varphi(t,x,y,\xi,\eta)}a(t)\hat{v}(\xi,\eta)\, d\xi\,d\eta,
\end{equation*}
where $\varphi$ is a quadratic form in $(x,y,\xi,\eta)$. Here $\hat{v}$ denotes the semiclassical Fourier transform:
\begin{equation*}
    \hat{v}(\xi,\eta)=\int e^{-\frac{i}{h}(x\xi+y\eta)}v(x,y)\,dx\,dy.
\end{equation*}
With this ansatz, we arrive at the expression
\begin{equation*}
    \begin{aligned}
    (hD_{t}+h^{2}D_{x}D_{y}-xy)u
    &=(2\pi h)^{-2}\iint\left[\frac{\partial\varphi}{\partial t}+\frac{\partial\varphi}{\partial x}\frac{\partial\varphi}{\partial y}-xy\right]
        e^{\frac{i\varphi}{h}}a(t)\hat{v}(\xi,\eta)\, d\xi\, d\eta\\
    &\qquad\quad-ih(2\pi h)^{-2}\iint\left[\frac{\partial a}{\partial t}+\frac{\partial^{2}\varphi}{\partial x\partial y}a\right]
        e^{\frac{i\varphi}{h}}\hat{v}(\xi,\eta)\, d\xi\, d\eta.
    \end{aligned}
\end{equation*}
Thus we wish to solve both the eikonal equation,
\begin{equation}\label{E:eikonalKx}
    \frac{\partial\varphi}{\partial t}+\frac{\partial\varphi}{\partial x}\frac{\partial\varphi}{\partial y}-xy=0,
\end{equation}
and also the transport equation,
\begin{equation}\label{E:transportKx}
    \frac{\partial a}{\partial t}+\frac{\partial^{2}\varphi}{\partial x\partial y}a=0.
\end{equation}
It is due to the special form of the operator $P$ that we are able to solve this in such a way that the amplitude $a$ depends only on $t$. Moreover,
we want $\varphi$ to satisfy
\begin{equation*}
    \varphi(0,x,y,\xi,\eta)=x\xi+y\eta,
\end{equation*}
and we want $a$ to solve
\begin{equation*}
    a(0)=1,
\end{equation*}
so that the initial condition is satisfied in the evolution equation (\ref{E:evoleqKx}).

We begin with a solution of the eikonal equation (\ref{E:eikonalKx}). The semiclassical Weyl symbol of $P$ is $$p(x,y,\xi,\eta)=\xi\eta-xy,$$ and we write the semiclassical Weyl symbol of the evolution equation (\ref{E:evoleqKx}) as
$$g(t,x,y,\tau,\xi,\eta)=\tau+p(x,y,\xi,\eta),$$ which is actually independent of $t$. The eikonal equation may then be written in the simple form
\begin{equation}\label{E:gdphi}
g(t,x,y,d_{t}\varphi,d_{x}\varphi,d_{y}\varphi)=0,
\end{equation}
whose solution, which we now sketch, is provided by Hamilton-Jacobi theory.

Hamilton's equations for $g$ are
\begin{equation*}
    \begin{cases}
    \dot{t}=\frac{\partial g}{\partial\tau}=1\\
    \dot{x}=\frac{\partial g}{\partial\xi}=\eta\\
    \dot{y}=\frac{\partial g}{\partial\eta}=\xi\\
    \dot{\tau}=-\frac{\partial g}{\partial t}=0\\
    \dot{\xi}=-\frac{\partial g}{\partial x}=y\\
    \dot{\eta}=-\frac{\partial g}{\partial y}=x.
    \end{cases}
\end{equation*}
Hence it is natural to identify the evolution parameter with time $t$. These equations may be easily solved to give the Hamilton flow of $g$:
\begin{equation*}
    \begin{cases}
    t=t\\
    x(t)=x_{0}\cosh t+\eta_{0}\sinh t\\
    y(t)=y_{0}\cosh t+\xi_{0}\sinh t\\
    \tau(t)=\tau_{0}\\
    \xi(t)=y_{0}\sinh t+\xi_{0}\cosh t\\
    \eta(t)=x_{0}\sinh t+\eta_{0}\cosh t.
    \end{cases}
\end{equation*}
One may think of this as just being the Hamilton flow of $p$, since the flow in the $(t,\tau)$ variables is trivial. In this point of view, we may write the Hamilton flow of $p$ in the matrix formulation:
\begin{equation*}
    \left(
    \begin{matrix}
    x(t)\\
    y(t)\\
    \xi(t)\\
    \eta(t)
    \end{matrix}
    \right)
    =
    \left(
    \begin{matrix}
    \cosh t&0&0&\sinh t\\
    0&\cosh t&\sinh t&0\\
    0&\sinh t&\cosh t&0\\
    \sinh t&0&0&\cosh t
    \end{matrix}
    \right)
    \left(
    \begin{matrix}
    x_{0}\\
    y_{0}\\
    \xi_{0}\\
    \eta_{0}
    \end{matrix}
    \right).
\end{equation*}
However, for the time being we take the point of view of the Hamilton flow of $g$. This Hamiltonian system (in a six-dimensional cotangent space) is completely integrable, since we have the three ($=\frac{6}{2}$) independent conserved quantities
\begin{equation*}
    \begin{cases}
    g(t,x,y,\tau,\xi,\eta)\\
    y^{2}-\xi^{2}\\
    \tau.
    \end{cases}
\end{equation*}
The fact that we have three conserved quantities corresponds to the fact that the flow is constrained to a Lagrangian submanifold; that is, a three-dimensional submanifold $\Lambda$ of the cotangent space such that the restriction of the symplectic form $\sigma$ to $\Lambda$ is zero (i.e., is \emph{isotropic}): $$\sigma|_{\Lambda}=0.$$

To construct a solution of the eikonal equation, the basic idea is to start with a two-dimensional isotropic submanifold $\Lambda^{\prime}$ and then to propagate $\Lambda^{\prime}$ with respect to the Hamilton flow of $g$, filling out the whole Lagrangian submanifold $\Lambda$. To this end, we let
\begin{equation*}
    \begin{aligned}
    \Lambda^{\prime}&=\{(0,x_{0},y_{0},\tau_{0},\xi_{0},\eta_{0});\, g(0,x_{0},y_{0},\tau_{0},\xi_{0},\eta_{0})=\tau_{0}+\xi_{0}\eta_{0}-x_{0}y_{0}=0\}\\
    &=\{(0,x_{0},y_{0},x_{0}y_{0}-\xi_{0}\eta_{0},\xi_{0},\eta_{0})\},
    \end{aligned}
\end{equation*}
which we think of as the ``submanifold of initial conditions'', and where we consider $(\xi_{0},\eta_{0})$ as universally fixed. Propagating along the Hamilton flow of $g$, we thus get the whole manifold $\Lambda$:
\begin{equation*}
    \Lambda=\{(t,x(t),y(t),x(t)y(t)-\xi(t)\eta(t),\xi(t),\eta(t))\}.
\end{equation*}

On the one hand, we may think of a trajectory along the Hamilton flow as being determined by the parameters $(t,x_{0},y_{0})$. On the other hand, we may think of it is as determined by the parameters $(t,x,y)$, since
\begin{equation*}
    \begin{cases}
    x_{0}=x\sech t-\eta_{0}\tanh t\\
    y_{0}=y\sech t-\xi_{0}\tanh t.
    \end{cases}
\end{equation*}
Hence, rewriting the variables $(\xi_{0},\eta_{0})$ as $(\xi,\eta)$, we may rewrite $\Lambda$ as
\begin{equation*}
    \begin{aligned}
    \Lambda&=\{(t,x,y,(xy-\xi\eta)\sech^{2}t-(x\xi+y\eta)\tanh t\sech t,\\
    &\qquad\qquad\qquad y\tanh t+\xi\sech t,\,\,x\tanh t+\eta\sech t)\}.
    \end{aligned}
\end{equation*}

To conclude the solution of the eikonal equation, we seek a function $\varphi$ such that $\Lambda$ is the graph of the gradient of $\varphi$, to agree with (\ref{E:gdphi}), and such that $\varphi(0,x,y,\xi,\eta)=x\xi+y\eta$. This is easily accomplished, and we thus have the phase:
\begin{equation*}
    \varphi(t,x,y,\xi,\eta)=(xy-\xi\eta)\tanh t+(x\xi+y\eta)\sech t.
\end{equation*}
And one may now check directly that this is the solution of the eikonal equation.

It is now easy to solve the transport equation (\ref{E:transportKx}):
\begin{equation*}
    a(t)=\sech t.
\end{equation*}
So we have the following expression for the solution operator $U_{t}$:
\begin{equation*}
    U_{t}v(x,y)=\sech t\iint \exp\left(\frac{i}{h}[(xy-\xi\eta)\tanh t+(x\xi+y\eta)\sech t]\right)\, \hat{v}(\xi,\eta)\, \frac{d\xi\, d\eta}{(2\pi h)^{2}}.
\end{equation*}
We may use the method of stationary phase (which, in this case, is \emph{exact}) to simplify this expression and get an integral in terms of $v$. We withhold the details for now, since this will be accomplished in the next section for a different but similar operator.

Also, it is known from the general theory (see, for example, \cite{R:GrigisSjostrand}) that $\varphi$ is a generating function of the canonical transformation, which in this case is the Hamilton flow of the symbol $p$. That is, the Hamilton flow of $p$ is given by
\begin{equation*}
    \left(\frac{\partial\varphi}{\partial\xi},\frac{\partial\varphi}{\partial\eta},\xi,\eta\right)
    \mapsto\left(x,y,\frac{\partial\varphi}{\partial x},\frac{\partial\varphi}{\partial y}\right).
\end{equation*}
This can also be checked directly, now that we have an explicit expression for the phase $\varphi$.

\vspace{12pt}

\section{The Gyrator Transform}\label{S:Gyrator}

The exact same method may be applied to the semiclassical differential operator
\begin{equation*}
    T_{1}=h(\crx\any+\cry\anx)=xy-h^{2}\frac{\partial^{2}}{\partial x\partial y}.
\end{equation*}
In the previous section, the solution of the evolution equation was exact, so the semiclassical parameter $h$ ultimately played no role. Hence in this section we simply let $h=1$. In this case, the solution to the evolution equation
\begin{equation*}
    \begin{cases}
    (D_{t}+T_{1})u=0\\
    u|_{t=0}=v
    \end{cases}
\end{equation*}
for $t\in (-\pi/2,\pi/2)$ is given by
\begin{equation}\label{E:gyraAA}
    u(t,x,y)=(2\pi)^{-2}|\sec t|\iint \exp\left(i[(x\xi+y\eta)\sec t -(xy+\xi\eta)\tan t]\right)\hat{v}(\xi,\eta)\, d\xi\, d\eta.
\end{equation}
It remains to extend the solution to all $t\in\Rbb$.

We may use the method of stationary phase (which is \emph{exact} in this case) to rewrite this integral as a double integral involving $v$. This simply amounts to a use of the following fact. Let $Q$ be a real, non-degenerate $n\times n$ symmetric matrix. Then the Fourier transform acts as follows (for details, see \cite{R:GrigisSjostrand}, p.21):
\begin{equation*}
    \mathcal{F}: e^{\frac{i}{2}\langle x,Qx\rangle}\mapsto (2\pi)^{\frac{n}{2}}e^{i\frac{\pi}{4}\text{sgn\,Q}}|\det Q|^{-\frac{1}{2}}e^{-\frac{i}{2}\langle \xi, Q^{-1}\xi\rangle}.
\end{equation*}
After some calculation, for $t\in (0,\pi/2)$ we arrive at the integral expression
\begin{equation}\label{E:gyraBB}
    u(t,x,y)=(2\pi)^{-1}|\sin t|^{-1}\iint v(a,b)\exp\left(i\frac{(xy+ab)\cos t-(ay+xb)}{\sin t}\right)\, da\,db.
\end{equation}
The right-hand side is known as ``the gyrator transform'' of $v$
\cite{R:RodrAlievCalvo}. The benefit of this expression is that we
may now extend the solution $u(t,x,y)$ to $t\in (0,\pi)$. For
$t\in (\pi/2,\pi)$ we may again use the method of stationary phase
to return to the expression (\ref{E:gyraAA}). There is no phase
shift, since in this example $\text{sgn\,Q}=0$.

Hence we have completely determined $u(t,x,y)$: for $t\in\Rbb\backslash\{(2k+1)\pi/2;\, k\in\mathbb{Z}\}$ it is given by (\ref{E:gyraAA}), and for $t\in\Rbb\backslash\{k\pi;\, k\in\mathbb{Z}\}$ it is given by (\ref{E:gyraBB}). This is analogous to the standard parametrization of the circle by four charts of graph coordinates. In fact, it is not only analogous, but intimately related; the exceptional points for (\ref{E:gyraAA}) (resp. (\ref{E:gyraBB})) are precisely those for which the Lagrangian submanifolds, swept out by the Hamilton flow of $T_{1}$, have degenerate projections onto the $(\xi,\eta)$ plane (resp. $(x,y)$ plane). (For more on this phenomenon, one may consult Duistermaat's beautiful review article \cite{R:DuistermaatUnfolding}.)

This unitary group has an important application when $t=\pm \pi/4$: it takes Hermite-Gaussian modes to Laguerre-Gaussian modes. For this we recall the definitions of the extended Wigner transform
\begin{equation*}
    \tilde{W}F(x,y)=\frac{1}{\sqrt{2\pi}}\int e^{ipy}F\left(\frac{x+p}{\sqrt{2}},\frac{x-p}{\sqrt{2}}\right)dp
\end{equation*}
and of the (renormalized) partial Fourier transform
\begin{equation*}
    \mathcal{F}_{2}F(x,y)=\frac{1}{\sqrt{2\pi}}\int e^{-ipy}F(x,p)dp.
\end{equation*}
Then one may check that
\begin{equation*}
    u(\pi/4,x,y)=\tilde{W}(\mathcal{F}_{2}v)(x,-y).
\end{equation*}

In the special case when the initial condition is $v=h_{mn}$, the $(m,n)^{th}$ Hermite-Gaussian mode, we have
\begin{equation*}
    u(\pi/4,x,y)=(-i)^{n}\tilde{W}(h_{nm})(x,y).
\end{equation*}
Moreover, by taking complex conjugates, we have
\begin{equation*}
    u(-\pi/4,x,y)=i^{n}\tilde{W}(h_{mn})(x,y).
\end{equation*}
And we recall that $\tilde{W}(h_{mn})$ is precisely the $(m,n)^{th}$ Laguerre-Gaussian mode \cite{R:VVLG}.

\vspace{12pt}

\section{The FIO Representation of a Non-Gaussian Transformation}\label{S:FIOnonGauss}

We now turn to the more difficult non-Gaussian transformations. The four non-Gaussian transformations used by Calvo and Pic\'{o}n may all be treated by the same methods, so we focus on the operator
\begin{equation*}
    \calt_{4}=\frac{1}{2}(\crx+\anx-\crx\crx\anx-\crx\anx\anx
    -\crx\cry\any-\anx\cry\any),
\end{equation*}
where the creation and annihilation operators are semiclassical. Moreover, we will make a slight simplification in order to remove inessential complications. Since the operator $\calt_{4}$ acts very simply in the $y$ variable, we may instead study the operator
\begin{equation*}
    \calt_{0}=\frac{1}{2}(\crx+\anx-\crx\crx\anx-\crx\anx\anx)
\end{equation*}
acting on functions of $x$ only. The following arguments may be repeated for $\calt_{4}$, but with some slight changes as outlined in Section~\ref{S:tilt4}.

We renormalize $\calt_{0}$ in order to have a semiclassical differential operator:
\begin{equation*}
    \begin{aligned}
    \tilt_{0}&=-2^{-1/2}h^{3/2}\calt_{0}\\
    &=\frac{1}{2}x(x^{2}+(hD_{x})^{2})-\frac{3}{2}hx-\frac{1}{2}ih^{2}D_{x}.
    \end{aligned}
\end{equation*}

The method we use to treat this operator is the same as in the previous sections, but there are some complications. The difficulties may be treated by the general theory, but the solution is not as explicit and exact. Here we wish to remain in the semiclassical setting, and we only expect an asymptotic solution to the corresponding evolution equation:
\begin{equation}\label{E:evoltilt}
    \begin{cases}
    \left(hD_{t}+\tilt_{0}\right)u=\mathcal{O}(h^{\infty})\\
    u|_{t=0}=v.
    \end{cases}
\end{equation}
We will allow the initial conditions $v$ to depend on $h$.
Moreover, for bounded times Duhamel's Principle shows that the
semiclassical propagator differs from the exact unitary propagator
by $\mathcal{O}(h^{\infty})$ (after choosing a self-adjoint
realization of the generator).

The semiclassical Weyl symbol of $\tilt_{0}$ is
\begin{equation*}
    p_{0}(x,\xi;h)=\frac{1}{2}x(x^{2}+\xi^{2})-\frac{3}{2}hx.
\end{equation*}
so then Hamilton's equations are
\begin{equation}\label{E:HamEqp0}
    \begin{cases}
    \dot{x}=x\xi\\
    \dot{\xi}=-x^{2}-\frac{1}{2}(x^{2}+\xi^{2})+\frac{3}{2}h.
    \end{cases}
\end{equation}
Here we have the conserved quantity
\begin{equation*}
    C_{0}=p_{0}(x,\xi;h).
\end{equation*}

Suppose first that $C_{0}\neq 0$, so that, in particular, $x(t)\neq 0$ for all $t$. Letting $$w=\frac{1}{x},$$ we have from (\ref{E:HamEqp0}) the following differential equation for $w$:
\begin{equation}\label{E:diffeqw}
    (\dot{w})^{2}=2C_{0}w^{3}+3hw^{2}-1.
\end{equation}
If $C_{0}\neq 0$, this is essentially the same as the differential equation for the famous Weierstrass $\wp$-function:
\begin{equation*}
    \left[\wp^{\prime}(z)\right]^{2}=4\left[\wp(z)\right]^{3}-g_{2}\wp(z)-g_{3}.
\end{equation*}
The two constants $g_{2}$ and $g_{3}$ are the so-called ``invariants''. We may then solve Hamilton's equations, giving the Hamilton flow in terms of the Weierstrass $\wp$-function.

To be explicit, for $C_{0}\neq 0$ we have
\begin{equation}\label{E:xflow}
    x(t)=\frac{1}{2}C_{0}\left(\wp(t+t_{0})-\frac{1}{4}h\right)^{-1}
\end{equation}
where $t_{0}$ is either an arbitrary real constant or an arbitrary real constant plus $\frac{1}{2}\omega_{1}$, the purely imaginary half-period of $\wp$ (see Appendix). Here $\wp$ is the Weierstrass $\wp$-function associated to the invariants
\begin{equation*}
    g_{2}=\frac{3}{4}h^{2} \qquad\text{and}\qquad g_{3}=\frac{1}{4}C_{0}^{2}-\frac{1}{8}h^{3}.
\end{equation*}

We then also have
\begin{equation}\label{E:xiflow}
    \begin{aligned}
    \xi(t)&=\frac{\dot{x}(t)}{x(t)}\\
    &=\frac{-\dot{\wp}(t+t_{0})}{\wp(t+t_{0})-\frac{1}{4}h}.
    \end{aligned}
\end{equation}
When $h=0$ and $C_{0}\neq 0$, $\xi(t)$ is always strictly decreasing, which follows simply from Hamilton's equations (\ref{E:HamEqp0}). However, when $h>0$ we have more complicated behavior, as shown in Figure \ref{F:hflow}. There is a pocket of radius $\sqrt{3h}=\sqrt{\frac{3}{2}}w_{0}$, where $w_{0}$ is in practice the radius of the laser beam's waist.

\vspace{12pt}

\begin{figure}
\begin{center}
\epsfig{file=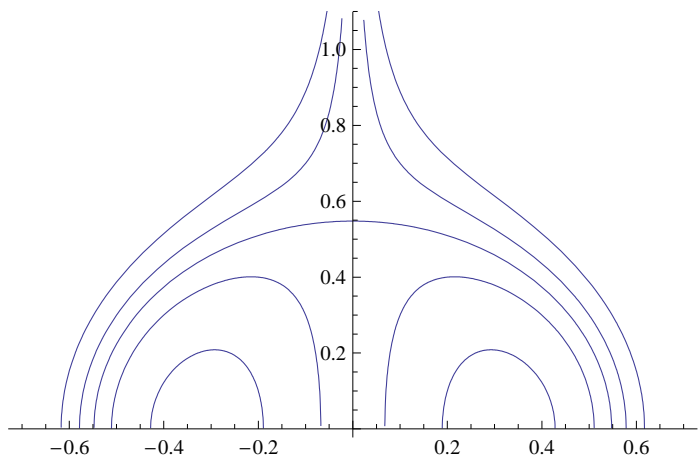,height=6cm} \caption{Hamilton flow lines in the
$(x,\xi)$ upper half plane in the case $h=1/10$, with variable $C_{0}$. The lower half
plane is obtained by symmetry.}\label{F:hflow}
\end{center}
\end{figure}

\vspace{12pt}

We have very different behavior when $C_{0}=0$. Depending on the initial conditions, we have one of the three following cases:
\begin{equation*}
    \begin{cases}
    x(t)= 0\quad\forall t\\
    \xi(t)=\sqrt{3h}\coth\left(\frac{\sqrt{3h}}{2}(t+t_{0})\right),
    \end{cases}
\end{equation*}
\begin{equation*}
    \begin{cases}
    x(t)= 0\quad\forall t\\
    \xi(t)=\sqrt{3h}\tanh\left(\frac{\sqrt{3h}}{2}(t+t_{0})\right),
    \end{cases}
\end{equation*}
or
\begin{equation*}
   \begin{cases}
   x(t)= \pm\sqrt{3h}\sech(\sqrt{3h}(t+t_{0}))\\
   \xi(t)=-\sqrt{3h}\tanh\left(\sqrt{3h}(t+t_{0})\right).
   \end{cases}
\end{equation*}
And we have the stationary points $(x,\xi)=(0,\pm\sqrt{3h})$ and $(x,\xi)=(\pm\sqrt{h},0)$.

\vspace{24pt}

Now we look for a solution of the equation (\ref{E:evoltilt}) of the form
\begin{equation}\label{E:T0ansatz}
    u=U_{t}v(x)=(2\pi h)^{-1}\int e^{\frac{i}{h}\varphi(t,x,\xi)}a(t,x,\xi;h)\hat{v}(\xi)\, d\xi.
\end{equation}
Here we are using the semiclassical Fourier transform, given by
\begin{equation*}
    \hat{v}(\xi)=\int e^{-\frac{i}{h}x\xi}v(x)\,dx.
\end{equation*}
We will even allow the phase $\varphi$ to depend on $h$, since we prefer to study the Weyl symbol of $\tilt_{0}$,
rather than the principal symbol only.

We take $\varphi$ to be a solution of the now $h$-dependent eikonal equation
\begin{equation}\label{E:eikonalT0}
    \frac{\partial\varphi}{\partial t}+p_{0}\left(x,\frac{\partial\varphi}{\partial x};h\right)=0.
\end{equation}
We may treat this using the same method as before, using Hamilton-Jacobi theory (see, for example, Theorem 5.5 of \cite{R:GrigisSjostrand}): given any $(x_{0},\xi_{0})\in\Rbb^{2}$, there exists a real-valued smooth function $\varphi(t,x;h)$ (with $(\xi_{0}$ considered as a parameter) defined in a neighborhood of $(0,x_{0})$ such that (\ref{E:eikonalT0}) is solved and such that
\begin{align*}
    \varphi(0,x)&=x\xi_{0},\\
    \frac{\partial\varphi}{\partial t}(0,x_{0})&=-\frac{1}{2}x_{0}(x_{0}^{2}+\xi_{0}^{2})+\frac{3}{2}hx_{0},  &&\text{and}\\
    \frac{\partial\varphi}{\partial x}(0,x_{0})&=\xi_{0}.
\end{align*}
However we have not yet managed to get a closed-form expression for the Hamilton flow (solely in terms of the initial conditions $x_{0}$ and $\xi_{0}$) or for the phase $\varphi$. Perhaps this is best suited for Laurent expansion methods and numerical methods, but we leave open the possibility that one may find closed-form solutions by brute force.

\vspace{12pt}

We next construct the amplitude $a$. For (\ref{E:T0ansatz}) to be a solution of the equation (\ref{E:evoltilt}), the amplitude $a$ must satisfy
\begin{equation*}
    (hD_{t}+\tilt_{0})(e^{\frac{i\varphi}{h}}a)=\mathcal{O}(h^{\infty}),
\end{equation*}
which is equivalent to
\begin{equation*}
    \left(\frac{\partial\varphi}{\partial t}+hD_{t}+e^{-\frac{i\varphi}{h}}\tilt_{0}e^{\frac{i\varphi}{h}}\right)a=\mathcal{O}(h^{\infty}).
\end{equation*}
And since $\varphi$ is a solution of the eikonal equation, we then have
\begin{equation}\label{E:ampl2}
    \left(hD_{t}+e^{-\frac{i\varphi}{h}}\tilt_{0}e^{\frac{i\varphi}{h}}-p_{0}\left(x,\frac{\partial\varphi}{\partial x};h\right)\right)a=\mathcal{O}(h^{\infty}).
\end{equation}
We simply compute
\begin{equation*}
    \begin{aligned}
    e^{-\frac{i\varphi}{h}}\tilt_{0}&e^{\frac{i\varphi}{h}}-p_{0}\left(x,\frac{\partial\varphi}{\partial x};h\right)\\
    &=-\frac{1}{2}ih\left(x\frac{\partial^{2}\varphi}{\partial x^{2}}+\frac{\partial\varphi}{\partial x}\right)+x\frac{\partial\varphi}{\partial x}hD_{x}+\frac{1}{2}x(hD_{x})^{2}-\frac{1}{2}ih^{2}D_{x}.
    \end{aligned}
\end{equation*}
If we let
\begin{equation*}
    V_{t}:=x\frac{\partial\varphi}{\partial x}\frac{\partial}{\partial x}
\end{equation*}
(which we recall depends on $h$), then (\ref{E:ampl2}) becomes
\begin{equation*}
    \left(\left(\frac{\partial}{\partial t}+\frac{1}{2}\dive V_{t}+V_{t}\right)+\frac{h}{2i}\left(\frac{\partial}{\partial x}\circ x\circ\frac{\partial}{\partial x}\right)\right)a=\mathcal{O}(h^{\infty}).
\end{equation*}
It is then natural to define
\begin{equation*}
    R_{1}:=\frac{\partial}{\partial t}+\frac{1}{2}\dive V_{t}+V_{t}
\end{equation*}
and
\begin{equation*}
    R_{2}:=\frac{1}{2i}\frac{\partial}{\partial x}\circ x\circ\frac{\partial}{\partial x}.
\end{equation*}

First considering a finite sum, we have
\begin{equation*}
    (R_{1}+hR_{2})\sum_{n=0}^{N}h^{n}a_{n}=R_{1}a_{0}+\sum_{n=1}^{N}h^{n}[R_{1}a_{n}+R_{2}a_{n-1}]+h^{N+1}R_{2}a_{N},
\end{equation*}
so  we wish to solve the transport equations
\begin{equation*}
    \begin{cases}
    R_{1}a_{0}=0\\
    R_{1}a_{n}+R_{2}a_{n-1}=0,\quad n\geq 1
    \end{cases}
\end{equation*}
with the initial conditions
\begin{equation*}
    \begin{cases}
    a_{0}(0,x,\xi)=1\\
    a_{n}(0,x,\xi)=0,\quad n\geq 1.
    \end{cases}
\end{equation*}

\vspace{12pt}

There is an elegant solution to the equation
\begin{equation}\label{E:a0eqn}
    \left(\frac{\partial}{\partial t}+\frac{1}{2}\dive V_{t}+V_{t}\right)a_{0}=0
\end{equation}
which gives interesting geometric information (see \cite{R:ZwSC} and Appendix A of \cite{R:HelffSjIII}). It is due to the fact that
\begin{equation*}
    V_{t}=\frac{\partial p_{0}}{\partial\xi}\left(x,\frac{\partial\varphi}{\partial x};h\right)\frac{\partial}{\partial x}.
\end{equation*}
We let
\begin{equation*}
    \Lambda_{t,\xi}:=\left\{\left(x,\frac{\partial\varphi}{\partial x}(t,x,\xi;h)\right)\right\}
\end{equation*}
and we denote the ($h$-dependent) Hamiltonian flow of $p_{0}$ by $\kappa$ (generated by the Hamiltonian vector field $H_{p_{0}}$), considered as taking
\begin{equation*}
    \kappa_{s,t}:\qquad\Lambda_{t-s,\xi}\rightarrow\Lambda_{t,\xi}.
\end{equation*}
We then note that
\begin{equation*}
    \frac{d}{ds}\kappa_{s,t}^{*}f\Big{|}_{s=0}=H_{p_{0}}\Big{|}_{\Lambda_{t,\xi}}f=V_{t}f
\end{equation*}
for $f\in C^{\infty}$.

Considering
\begin{equation*}
    a_{0}(t,x,\xi)|dx|^{1/2}
\end{equation*}
as a half-density on $\Lambda_{t,\xi}$, (\ref{E:a0eqn}) then becomes
\begin{equation*}
    \frac{d}{dt}\kappa_{t}^{*}(a_{0}|dx|^{1/2})=\left(\frac{\partial}{\partial t}+\mathcal{L}_{V_{t}}\right)(a_{0}|dx|^{1/2})=0
\end{equation*}
where $\mathcal{L}_{V_{t}}$ denotes the Lie derivative. That is, the amplitude, interpreted as a half-density, is invariant under the flow. This is the same as
\begin{equation*}
    \kappa_{t}^{*}(a_{0}(t,x,\xi)|dx|^{1/2}\Big{|}_{\Lambda_{t,\xi}})=|dx|^{1/2}\Big{|}_{\Lambda_{0,\xi}}.
\end{equation*}
Hence
\begin{equation*}
    \kappa_{t}^{*}a_{0}=|\partial\kappa_{t}|^{-1/2},
\end{equation*}
where $\partial\kappa_{t}$ denotes the differential of the mapping $\kappa_{t}$.

Now we recall that the mapping $\kappa_{t}$ may be described in terms of its generating function $\varphi$, so that its inverse is given by
\begin{equation*}
    \kappa_{t}^{-1}:\quad \left(x,\frac{\partial\varphi}{\partial x}(t,x,\xi;h)\right)
    \rightarrow \left(\frac{\partial\varphi}{\partial \xi}(t,x,\xi;h),\xi\right).
\end{equation*}
Hence
\begin{equation*}
    \partial(\kappa_{t}^{-1}\Big{|}_{\Lambda_{t,\xi}})
    =\frac{\partial^{2}\varphi}{\partial x\partial \xi}.
\end{equation*}
So finally we see that
\begin{equation*}
    a_{0}(t,x,\xi)=\left(
        \frac{\partial^{2}\varphi}{\partial x\partial \xi}\right)^{1/2}.
\end{equation*}

\vspace{24pt}

\begin{remark}
\emph{
The same argument works for more general operators \cite{R:ZwSC}, \cite{R:HelffSjIII}. In particular, the simple operators in Sections \ref{S:WarmUp} and \ref{S:Gyrator} fit into this framework. In Section \ref{S:WarmUp} we had the phase
\begin{equation*}
    \varphi(t,x,y,\xi,\eta)=(xy-\xi\eta)\tanh t+(x\xi+y\eta)\sech t,
\end{equation*}
so that
\begin{equation*}
    \det
    \begin{pmatrix}
        \frac{\partial^{2}\varphi}{\partial x\partial \xi}&\frac{\partial^{2}\varphi}{\partial x\partial \eta}\\
        \frac{\partial^{2}\varphi}{\partial y\partial \xi}&\frac{\partial^{2}\varphi}{\partial y\partial \eta}
    \end{pmatrix}
    =\sech^{2}t.
\end{equation*}
And in Section \ref{S:Gyrator} we had the phase
\begin{equation*}
    \varphi(t,x,y,\xi,\eta)=(x\xi+y\eta)\sec t-(xy+\xi\eta)\tan t,
\end{equation*}
so that
\begin{equation*}
    \det
    \begin{pmatrix}
        \frac{\partial^{2}\varphi}{\partial x\partial \xi}&\frac{\partial^{2}\varphi}{\partial x\partial \eta}\\
        \frac{\partial^{2}\varphi}{\partial y\partial \xi}&\frac{\partial^{2}\varphi}{\partial y\partial \eta}
    \end{pmatrix}
    =\sec^{2}t.
\end{equation*}
}
\end{remark}

\vspace{24pt}

We may also solve the higher-order transport equations to construct the full amplitude
\begin{equation*}
    a\sim\sum_{n=0}^{\infty}h^{n}a_{n}
\end{equation*}
in the sense of Borel summation.

\vspace{12pt}

To rigorously solve (\ref{E:evoltilt}) one must control the error. For this one uses cutoff functions in phase space and restricts to initial conditions $v$ that are appropriately localized in phase space (for a textbook presentation, see \cite{R:ZwSC}). For example, finite sums of Hermite-Gaussian modes are localized to the origin in phase space, and so are localized to the flow-invariant disc. So fortunately we have good behavior for the physically most relevant initial conditions.

\vspace{12pt}

\section{The Operator $\tilt_{4}$}\label{S:tilt4}

For reference, in this section we give the Hamilton flow for the operator $\tilt_{4}$. The methods of the previous section may be applied, with some slight modifications due to the additional parameters.

The semiclassical Weyl symbol of $\tilt_{4}$ is
\begin{equation*}
    p_{4}(x,y,\xi,\eta;h)=\frac{1}{2}x(x^{2}+y^{2}+\xi^{2}+\eta^{2})-\frac{5}{2}hx,
\end{equation*}
so then Hamilton's equations are
\begin{equation}\label{E:HamEqp4}
    \begin{cases}
    \dot{x}=x\xi\\
    \dot{y}=x\eta\\
    \dot{\xi}=-x^{2}-\frac{1}{2}(x^{2}+y^{2}+\xi^{2}+\eta^{2})+\frac{5}{2}h\\
    \dot{\eta}=-xy.
    \end{cases}
\end{equation}
Here we have the two conserved quantities
\begin{equation*}
    \begin{cases}
    C_{0}=p_{4}(x,y,\xi,\eta)\\
    C_{1}^{2}=y^{2}+\eta^{2}.
    \end{cases}
\end{equation*}

For $C_{0}\neq 0$ we have the solution
\begin{equation}\label{E:xflowDim4}
    x(t)=\frac{1}{2}C_{0}\left(\wp(t+t_{0})+\frac{1}{12}(C_{1}^{2}-5h)\right)^{-1}
\end{equation}
where $t_{0}$ is either an arbitrary real constant or an arbitrary real constant plus $\frac{1}{2}\omega_{1}$, the purely imaginary half-period of $\wp$ (see Appendix). Here $\wp$ is the Weierstrass $\wp$-function associated to the invariants
\begin{equation*}
    g_{2}=\frac{1}{12}(C_{1}^{2}-5h)^{2} \qquad\text{and}\qquad g_{3}=\frac{1}{216}(C_{1}^{2}-5h)^{3}+\frac{1}{4}C_{0}^{2}.
\end{equation*}

We then also have
\begin{equation}\label{E:xiflowDim4}
    \xi(t)=\frac{-\dot{\wp}(t+t_{0})}{\wp(t+t_{0})+\frac{1}{12}(C_{1}^{2}-5h)},
\end{equation}

As for $y$, we have the equation
\begin{equation*}
    \frac{1}{x}\frac{d}{dt}\left(\frac{1}{x}\dot{y}\right)+y=0
\end{equation*}
where $x$ is as in (\ref{E:xflowDim4}). The solution is given by
\begin{equation*}
    y(t)=y_{0}\cos\int^{t}_{0}x(s)\, ds +\eta_{0}\sin\int^{t}_{0}x(s)\, ds
\end{equation*}
and hence
\begin{equation*}
    \eta(t)=-y_{0}\sin\int^{t}_{0}x(s)\, ds  +\eta_{0}\cos\int^{t}_{0}x(s)\, ds
\end{equation*}
where of course $$y_{0}^{2}+\eta_{0}^{2}=C_{1}^{2}.$$

We have very different behavior when $C_{0}=0$, which breaks into multiple cases. If $C_{1}^{2}\geq 5h$, then $x(t)=0$ for all $t$. We then have $y(t)\equiv y_{0}$ and $\eta(t)\equiv \eta_{0}$. As for $\xi$, if $C_{1}^{2}>5h$, then
\begin{equation*}
    \xi(t)=-\sqrt{C_{1}^{2}-5h}\,\,\tan\left(\frac{1}{2}\sqrt{C_{1}^{2}-5h}\,\,(t+t_{0})\right),
\end{equation*}
where $t_{0}$ is an arbitrary constant. And if $C_{1}^{2}=5h$, then either $\xi\equiv 0$ or $\xi(t)=\frac{2}{t+t_{0}}$.

On the other hand, if $C_{1}^{2}<5h$, we have one of the following three cases, depending on the initial conditions. Either
\begin{equation*}
    \begin{aligned}
    x(t)&=\pm\sqrt{5h-C_{1}^{2}}\,\,\sech\left(\sqrt{5h-C_{1}^{2}}\,\,(t+t_{0})\right)\\
    \xi(t)&=-\sqrt{5h-C_{1}^{2}}\,\,\tanh\left(\sqrt{5h-C_{1}^{2}}\,\,(t+t_{0})\right)\\
    y(t)&=y_{0}\cos\int^{t}_{0}x(s)\, ds +\eta_{0}\sin\int^{t}_{0}x(s)\, ds\\
    \eta(t)&=-y_{0}\sin\int^{t}_{0}x(s)\, ds  +\eta_{0}\cos\int^{t}_{0}x(s)\, ds,
    \end{aligned}
\end{equation*}
or, in the case where $x_{0}=0$,
\begin{equation*}
    \begin{aligned}
    x(t)&\equiv 0,\qquad y(t)\equiv y_{0},\qquad \eta(t)\equiv \eta_{0},\qquad\text{and}\\
    \xi(t)&=\sqrt{5h-C_{1}^{2}}\,\,\tanh\left(\frac{1}{2}\sqrt{5h-C_{1}^{2}}\,\,(t+t_{0})\right),
    \end{aligned}
\end{equation*}
or
\begin{equation*}
    \begin{aligned}
    x(t)&\equiv 0,\qquad y(t)\equiv y_{0},\qquad \eta(t)\equiv \eta_{0},\qquad\text{and}\\
    \xi(t)&=\sqrt{5h-C_{1}^{2}}\,\,\coth\left(\frac{1}{2}\sqrt{5h-C_{1}^{2}}\,\,(t+t_{0})\right).
    \end{aligned}
\end{equation*}

And we have the stationary points
\begin{equation*}
    (x,y,\xi,\eta)=(\pm\sqrt{\frac{5}{3}h},0,0,0)
\end{equation*}
and, for $y_{0}^{2}+\eta_{0}^{2}=C_{1}^{2}\leq 5h$,
\begin{equation*}
    (x,y,\xi,\eta)=(0,y_{0},\pm\sqrt{5h-C_{1}^{2}},\eta_{0}).
\end{equation*}

\vspace{12pt}

\section{Egorov's Theorem}\label{S:Egorov}

When dealing with the semiclassical quantization of a quadratic polynomial, one may use the metaplectic representation. This is the method discussed, for example, in the work of Calvo and Pic\'{o}n \cite{R:CalvoPicon}. Here we will simply define the metaplectic representation and note the connection with Egorov's theorem.

Following Folland \cite{R:Folland}, we write the Schr\"{o}dinger representation of the Heisenberg group as
\begin{equation*}
    \rho(p,q,t)=e^{2\pi i(pD+qX+tI)},
\end{equation*}
and we write the metaplectic representation as $\mu$. The metaplectic representation is sometimes described in the language of representation theory as follows. Let $\mathcal{T}$ denote the group of automorphisms of the Heisenberg group $\mathbb{H}_{n}$ that leave the center pointwise fixed. If $T\in\mathcal{T}$, then the composition $\rho\circ T$ is a new irreducible unitary representation of the Heisenberg group, nontrivial on the center, so the Stone-von Neumann theorem \cite{R:Folland} says that there exists a unitary operator $\mu(T)$ on $L^{2}(\Rbb^{n})$ such that
\begin{equation*}
    \mu(T)\rho(X)\mu(T)^{-1}=\rho\circ T(X),\qquad X\in\mathbb{H}_{n}.
\end{equation*}

Let us compare this to (one version of) Egorov's Theorem:
\begin{equation*}
    e^{it\tilt_{4}/h}\text{Op}^{W}_{h}(q)e^{-it\tilt_{4}/h}=\text{Op}^{W}_{h}(q_{t})
\end{equation*}
with $q_{t}=q\circ\kappa_{t}+\mathcal{O}(h^{2})$. Since we are
dealing with unitary operators and the Weyl quantization, we have
an error of order $\mathcal{O}(h^{2})$ rather than the more usual
$\mathcal{O}(h)$ (see Appendix~A of \cite{R:HelffSjIII} or
Section~2 of \cite{R:HitrSjI}). This is in fact one of the main
reasons for our using $h$-dependent canonical transformations.
Here $\kappa_{t}=\exp(tH_{p_{4}})$ is the ($h$-dependent) Hamilton
flow associated to $\tilt_{4}$. Egorov's theorem is a way of
justifying the intuition that ``the Fourier integral operator
$e^{itH/h}$ quantizes the Hamilton flow of $H$''. And there is a
wider class of Fourier integral operators that may be considered
as quantizations of canonical transformations
\cite{R:GrigisSjostrand}.

\vspace{12pt}

This has a useful consequence for optics, where the Wigner transform is a widely used tool for studying phase space properties of functions. For this we define the standard $n$-dimensional semiclassical Wigner transform by
\begin{equation*}
    W(f,g)(x,\xi)=\int e^{-\frac{i\xi p}{h}}f\left(x+\frac{1}{2}p\right)\overline{g\left(x-\frac{1}{2}p\right)}\, dp.
\end{equation*}
One may check that for any symbol $\sigma$ one has
\begin{equation*}
    \langle\Op(\sigma)f|g\rangle=\iint\sigma(x,\xi)W(f,g)(x,\xi)\,dx\,d\xi.
\end{equation*}
Then, for example, using $U_{t}$, the semiclassical approximate
propagator for $\tilt_{4}$ (with Weyl symbol $p_{4}$, having the
Hamilton flow $\kappa_{t}$), one has
\begin{equation*}
    \begin{aligned}
    \langle\Op(\sigma)U_{t}f|U_{t}g\rangle
    &=\langle e^{\frac{it\tilt_{4}}{h}}\Op(\sigma)e^{-\frac{it\tilt_{4}}{h}}f|g\rangle+\mathcal{O}(h^{\infty})\\
    &=\langle\Op(\sigma_{t})f|g\rangle+\mathcal{O}(h^{\infty}),
    \end{aligned}
\end{equation*}
where $\sigma_{t}=\sigma\circ\kappa_{t}+\mathcal{O}(h^{2})$. Hence
\begin{equation*}
    \begin{aligned}
    \iint\sigma(x,\xi)W(U_{t}f,U_{t}g)(x,\xi)\,dx\,d\xi
    &=\iint\sigma(\kappa_{t}(x,\xi))W(f,g)(x,\xi)\,dx\,d\xi+\mathcal{O}(h^{2})\\
    &=\iint\sigma(x,\xi)W(f,g)(\kappa_{-t}(x,\xi))\,dx\,d\xi+\mathcal{O}(h^{2}).
    \end{aligned}
\end{equation*}
Since this is true for all symbols $\sigma$, we thus have
\begin{equation*}
    W(U_{t}f,U_{t}g)=W(f,g)\circ\kappa_{-t}+\mathcal{O}(h^{2}).
\end{equation*}
If we were dealing with a metaplectic operator, the transformation $\kappa$ would then be a linear symplectic transformation, and the identity would be exact.

\vspace{12pt}

\section{Additional Remarks}\label{S:Concl}

     \begin{remark} \emph{
     So far we have concerned ourselves with the manipulation of Hermite-Gaussian modes: those two-dimensional modes generated by applying the $x$ and $y$ creation operators $\crx$ and $\cry$ to the fundamental mode $\pi^{-1/2}e^{-\frac{1}{2}(x^{2}+y^{2})}$ (one may alternatively use the semiclassical version). But one may also just as easily consider the manipulation of Laguerre-Gaussian modes, generated by applying the creation operators
     \begin{equation*}
        \hat{A}_{\pm}^{\dag}=2^{-1/2}(\crx\pm i\cry)
     \end{equation*}
     to the fundamental mode. It is no more difficult to use the above methods, since the extended Wigner transform intertwines the two classes of creation operators \cite{R:VVLG}.}
     \end{remark}

    \vspace{12pt}

    \begin{remark}
    \emph{
    We have seen that the Hamilton flows of
    $\tilt_{4}$, $\tilt_{5}$, $\tilt_{6}$, and $\tilt_{7}$ blow up in
    finite time. One may now ask: what are the consequences for the
    propagators themselves? As noted by Maciej Zworski \cite{R:ZwPersC},
    the propagators (the non-Gaussian transformations)
    may cause certain initial conditions to develop singularities in
    finite time, since this is precisely what happens for simpler
    operators.
    }

    \emph{
    We consider the following simplification of $\tilt_{5}$:
    \begin{equation*}
        Q_{1}=\frac{1}{2}(x^{2}D_{x}+D_{x}\circ x^{2}).
    \end{equation*}
    One may explicitly find the deficiency subspaces, the Hamilton flow, and the propagator (which may be viewed as a Fourier integral operator). And there exist smooth $L^{2}(\Rbb)$ initial conditions, namely,
    \begin{equation*}
        u_{0}(x)=\langle x \rangle^{-\alpha},\qquad \frac{1}{2}<\alpha<1,
    \end{equation*}
    which develop singularities, traveling along the Hamilton flow.
    }

    \emph{
    For an example closer yet to $\tilt_{5}$, one may consider
    \begin{equation*}
        \begin{aligned}
        Q_{2}&=\frac{1}{2}(x^{2}hD_{x}+hD_{x}\circ x^{2})-5h^{2}D_{x}\\
        &=x^{2}hD_{x}-ihx-5h^{2}D_{x}.
        \end{aligned}
    \end{equation*}
    For reference, we note that the propagator is
    \begin{equation*}
        e^{-\frac{itQ_{2}}{h}}u_{0}(t,x)=a(x\sinh(at)+a\cosh(at))^{-1}u_{0}\left(a\left(\frac{x\cosh(at)+a\sinh(at)}{x\sinh(at)+a\cosh(at)}\right)\right)
    \end{equation*}
    where $a=\sqrt{5h}$.
    }

    \emph{To study creation of singularities in general, one might use the methods of Jared Wunsch \cite{R:WunschGrow}.}

    \end{remark}

\vspace{12pt}

\section*{Acknowledgements}

The author benefitted from conversations with Hamid Hezari at the
2008 Clay Math Institute Summer School, with Michael Hitrik, James
Ralston, Gregory Eskin, and William Duke at UCLA, and with Maciej
Zworski and Jared Wunsch at Berkeley/MSRI. The author also thanks
Gabriel F. Calvo and Antonio Pic\'{o}n for contacting him
personally and suggesting this problem.

\vspace{12pt}

\appendix
\section{The Weierstrass $\wp$-Function}

In Section \ref{S:FIOnonGauss} the Weierstrass $\wp$-function
played a central role, as elliptic functions were found to
parametrize the ($h$-dependent) Hamilton flow of
\begin{equation*}
    p_{0}(x,\xi;h)=\frac{1}{2}x(x^{2}+\xi^{2})-\frac{3}{2}hx.
\end{equation*}
Here we briefly describe some useful facts about the $\wp$-function, in the context of our problem. In what follows, we occasionally state (but do not prove) standard results, taken from the textbooks of Ahlfors \cite{R:Ahlfors} and Koblitz \cite{R:Koblitz}.

The Weierstrass $\wp$-function is a doubly periodic meromorphic function on the complex plane, say, with periods $\omega_{1},\omega_{2}\in\Cbb$, with a double pole at each point of the period lattice, including the origin. And it is a solution of the differential equation
\begin{equation}\label{E:WPDiffeq}
    [\wp^{\prime}(z)]^{2}=4[\wp(z)]^{3}-g_{2}\wp(z)-g_{3},
\end{equation}
where the ``invariants'' $g_{2}$ and $g_{3}$ characterize $\wp$ just as well as $\omega_{1}$ and $\omega_{2}$. In our case,
\begin{equation*}
    g_{2}=\frac{3}{4}h^{2}\quad\text{and}\quad g_{3}=\frac{1}{4}C_{0}^{2}-\frac{1}{8}h^{3}.
\end{equation*}
Here $C_{0}:=p_{0}(x,\xi;h)$ is conserved under the flow. Since the differential equation has constant coefficients, $\wp(z+z_{0})$ is also a solution, for any $z_{0}\in\Cbb$.

One may factorize the cubic polynomial (\ref{E:WPDiffeq}):
\begin{equation*}
    \wp^{\prime}(z)^{2}=4(\wp(z)-e_{1})(\wp(z)-e_{2})(\wp(z)-e_{3}).
\end{equation*}
It is a familiar fact that the roots $e_{1},e_{2},e_{3}$ of $4x^{3}-g_{2}x-g_{3}$ are distinct if and only if the discriminant is nonzero:
\begin{equation*}
    \Delta:=g_{2}^{3}-27 g_{3}^{2}\neq 0.
\end{equation*}
Moreover, all the roots $e_{1}$, $e_{2}$, $e_{3}$ are real if and only if $g_{2}$ and $g_{3}$ are real and $\Delta>0$. In our case,
\begin{equation*}
    \Delta=\frac{27}{16}C_{0}^{2}(h^{3}-C_{0}^{2}).
\end{equation*}
So if we ignore the trivial case $C_{0}=0$ (see Section \ref{S:FIOnonGauss}), we have that $e_{1},e_{2},e_{3}$ are distinct if and only if $C_{0}^{2}\neq h^{3}$, and moreover $e_{1},e_{2},e_{3}$ are real if and only if $C_{0}^{2} < h^{3}$.

Suppose that we are in the case $\Delta>0$. We then order the $e_{i}$ such that $e_{1}<e_{3}<e_{2}$. Then one may choose the periods $\omega_{1}$ and $\omega_{2}$ of $\wp$ to be given by
\begin{equation}\label{E:periodintOne}
    \frac{1}{2}\omega_{1}=i\int_{-\infty}^{e_{1}}\frac{dt}{\sqrt{g_{3}+g_{2}t-4t^{3}}}
\end{equation}
and
\begin{equation}\label{E:periodintTwo}
    \frac{1}{2}\omega_{2}=\int_{e_{2}}^{\infty}\frac{dt}{\sqrt{4t^{3}-g_{2}t-g_{3}}},
\end{equation}
where we take the positive branch of the square root and integrate along the real axis. Hence the fundamental domain of $\wp$ in the complex plane is a rectangle with its opposite edges identified.

In our case, in Section \ref{S:FIOnonGauss}, it is clear that the unbounded flow lines are described by restricting $\wp$ to the real line, since the poles lie on the real line. Also, it is clear that the nondegenerate flow loops are described by restricting $\wp$ to a line of the form
\begin{equation*}
    L_{\alpha}=\{t+\alpha;\,t\in\Rbb\}, \quad\text{for some $\alpha\in\Cbb$.}
\end{equation*}

To determine the value of $\alpha$, we make use of the addition formula
\begin{equation*}
    \wp(z+u)=-\wp(z)-\wp(u)+\frac{1}{4}\left(\frac{\wp^{\prime}(z)-\wp^{\prime}(u)}{\wp(z)-\wp(u)}\right)^{2}.
\end{equation*}
For $\wp(t+\alpha)$ to be real for all $t\in\Rbb$, we then only need $\alpha$ to satisfy $\wp(\alpha)\in\Rbb$ and $\wp^{\prime}(\alpha)\in\Rbb$.
But it is well known that
\begin{equation*}
    e_{1}=\wp\left(\frac{\omega_{1}}{2}\right),\quad e_{2}=\wp\left(\frac{\omega_{2}}{2}\right), \quad e_{3}=\wp\left(\frac{\omega_{1}+\omega_{2}}{2}\right),
\end{equation*}
and that
\begin{equation*}
    \wp^{\prime}\left(\frac{\omega_{1}}{2}\right)=0,\quad \wp^{\prime}\left(\frac{\omega_{2}}{2}\right)=0,\quad\text{and}\quad \wp^{\prime}\left(\frac{\omega_{1}+\omega_{2}}{2}\right)=0.
\end{equation*}
Hence, in light of (\ref{E:periodintOne}) and (\ref{E:periodintTwo}), we may simply take $\alpha=\frac{1}{2}\omega_{1}$.

This shows that the nondegenerate flow loops are parametrized by
\begin{equation*}
    \Rbb\ni t\mapsto (x(t),\xi(t)),
\end{equation*}
with
\begin{equation*}
    x(t)=\frac{1}{2}C_{0}\left(\wp\left(t+\frac{\omega_{1}}{2}\right)-\frac{1}{4}h\right)^{-1}
\end{equation*}
and
\begin{equation*}
    \xi(t)=\frac{-\wp^{\prime}\left(t+\frac{\omega_{1}}{2}\right)}{\wp\left(t+\frac{\omega_{1}}{2}\right)-\frac{1}{4}h}.
\end{equation*}

This can easily be seen in the figures. The fundamental domain of $\wp$ is pictured in Figure \ref{F:WPdomain}. As $t$ traverses the horizonal line ($a$), $(x(t),\xi(t))$ traces out the curve ($a$) in Figure \ref{F:Ellcurve}. And as $t+\frac{\omega_{1}}{2}$ traverses the horizontal line ($b$), $(x(t),\xi(t))$ traces out the curve ($b$) in Figure \ref{F:Ellcurve}. The parameters in Figure \ref{F:Ellcurve} are $h=0.1$ and $C_{0}=0.025$. Note in particular that $C_{0}^{2}<h^{3}$. When $C_{0}^{2}=h^{3}$, the loop reduces to a point, although there is still an unbounded component. And there is only an unbounded component when $C_{0}^{2}>h^{3}$.

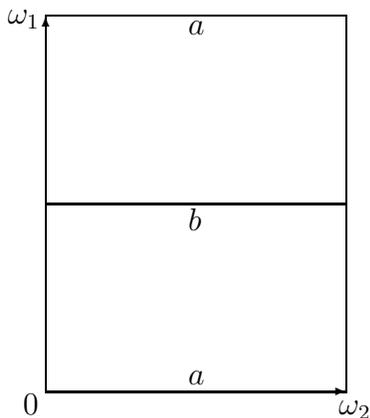
\begin{figure}
\setlength{\unitlength}{5cm}
\begin{picture}(.8,1)
\put(0,0){\vector(0,1){1}} \put(0,0){\vector(1,0){.8}}
\put(.8,0){\line(0,1){1}} \put(0,1){\line(1,0){.8}}
\put(0,.5){\line(1,0){.8}}
\put(.38,.95){$a$}
\put(.38,.43){$b$}
\put(.38,.02){$a$}
\put(.78,-.06){$\omega_{2}$}
\put(-.1,.98){$\omega_{1}$}
\put(-.06,-.06){$0$}
\end{picture}
\caption{The fundamental domain of $\wp$. When restricted to either the line ($a$) or the line ($b$), $\wp$ is real-valued.}\label{F:WPdomain}
\end{figure}

\begin{figure}
\epsfig{file=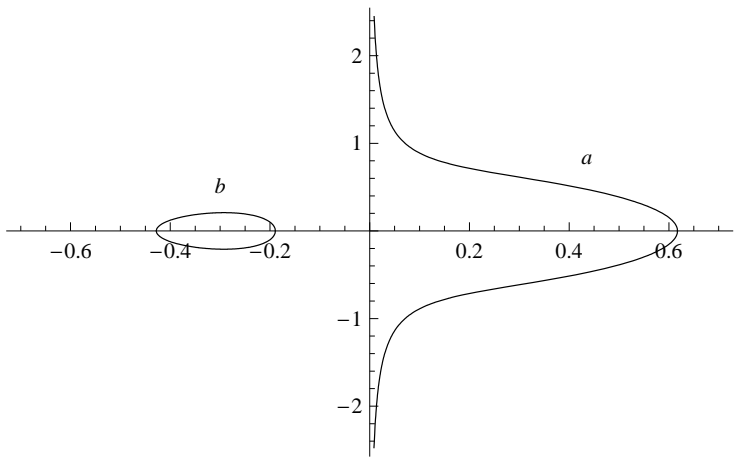,height=6cm}
\caption{The elliptic curve parametrized by $(x(t),\xi(t))$. The parameters for the curve are $h=0.1$ and $C_{0}=0.025$.}\label{F:Ellcurve}
\end{figure}


\end{document}